\documentclass[11pt,reqno,a4paper]{amsart}
\pdfoutput=1
\newif\ifshort
\shortfalse
\usepackage[stable]{footmisc}
\usepackage[backgroundcolor=orange!50,textsize=scriptsize]{todonotes}
\def\keywords{\smallskip\noindent\textsc{Key Words. }}

\usepackage{amsfonts,amstext,amsmath,amssymb,stmaryrd,mathrsfs,calc,amsthm,mathtools,bold-extra}
\let\oldsec\S
\usepackage[small,basic]{complexity}
\def\isnum#1{%
  \if!\ifnum9<1#1!\else_\fi
    \ComplexityFont{#1}\else#1\fi}
\renewcommand{\EXP}[1][]{{\isnum{#1}\ComplexityFont{EXPTIME}}}
\renewcommand{\P}{\ComplexityFont{PTIME}}
\newclass{\TOWER}{TOWER}
\renewcommand{\S}{\oldsec}

\newcommand{\tuple}[1]{\left(#1\right)}

\newcommand{\avg}[1]{\mathsf{avg}(#1)}
\newcommand{\tsum}{\textstyle{\sum}}
\newcommand{\vsp}[1]{\mathsf{span}(#1)}
\newcommand{\cone}[1]{\mathsf{cone}(#1)}

\newcommand{\eqby}[1]{\stackrel{\text{{\tiny{#1}}}}{=}}
\newcommand{\eqdef}{\eqby{def}}
\renewcommand{\vec}[1]{\mathbf{#1}}

\usepackage{thmtools,thm-restate}
\renewcommand{\thmcontinues}[1]{%
    \hyperref[#1]{continued}%
}%
\declaretheorem[numberwithin=section]{theorem}
\declaretheorem[sibling=theorem]{proposition}
\declaretheorem[sibling=theorem]{lemma}
\declaretheorem[sibling=theorem]{corollary}

\declaretheorem[sibling=theorem]{fact}

\declaretheorem[sibling=theorem,style=remark,qed=\qedsymbol]{example}
\declaretheorem[sibling=theorem]{claim}
\declaretheoremstyle[
spaceabove=6pt,spacebelow=6pt,
headfont=\normalfont\bfseries,
notefont=\mdseries, notebraces={}{},
bodyfont=\normalfont,
postheadspace=1em,
numbered=no
]{problem}

\usepackage{subfig,tikz}
\usetikzlibrary{positioning}
\usetikzlibrary{arrows}
\usetikzlibrary{automata}
\usetikzlibrary{shapes.geometric}
\tikzstyle{triangle}=[rounded corners=.5,regular polygon,regular polygon
  sides=3,minimum size=.5cm,draw=gray!90,inner sep=1pt,fill=gray!20,very thick]
\tikzstyle{square}=[rounded corners=.5,regular polygon,regular polygon
  sides=4,minimum size=.5cm,draw=gray!90,inner sep=1pt,fill=gray!20,very thick]
\tikzstyle{every node}=[font=\small]
\tikzstyle{every edge}=[draw,->,>=stealth',shorten >=1pt,semithick]
\tikzstyle{accepting}=[accepting by arrow]
\tikzstyle{initial}=[initial by arrow,initial text=]
\makeatletter
\def\NAT@spacechar{~}%
\def\@listi{\leftmargin\leftmargini
               \topsep 3\p@ \@plus\p@ \@minus\p@
               \parsep 2\p@ \@plus\p@ \@minus\p@
               \itemsep \parsep}
\makeatother
\usepackage{hyperref}

\providecommand{\urlstyle}[1]{}
\providecommand{\doi}[1]{\href{http://dx.doi.org/#1}{\nolinkurl{doi:#1}}}
\newcommand{\appref}{\autoref}
\usepackage[numbers]{natbib}
\renewcommand{\cite}{\citep}

\usepackage{enumitem}
\begin{document}
\renewcommand{\sectionautorefname}{Section}
\renewcommand{\subsectionautorefname}{Section}
\renewcommand{\subsubsectionautorefname}[1]{\S}
\title{Fixed-Dimensional Energy Games are in Pseudo-Polynomial Time}
\thanks{Work partially supported by ANR grant~11-BS02-001-01
  \textsc{ReacHard}, the Leverhulme Trust Visiting
  Professorship~1-2014-041, and the EPSRC grant~EP/M011801/1.}
\author[M.~Jurdzi\'nski]{Marcin Jurdzi\'nski}
\author[R.~Lazi\'c]{Ranko Lazi\'c} \address{DIMAP, Department of
  Computer Science, University of Warwick, UK}
\email{\{mju,lazic\}@dcs.warwick.ac.uk} \author[S.~Schmitz]{Sylvain
  Schmitz} \address{LSV, ENS Cachan \& CNRS \& INRIA, France}
\email{schmitz@lsv.ens-cachan.fr}

\begin{abstract}
We generalise the hyperplane separation technique (Chatterjee and
Velner, 2013) from multi-dimensional mean-payoff to energy games, and
achieve an algorithm for solving the latter whose running time is
exponential only in the dimension, but not in the number of vertices
of the game graph.  This answers an open question whether energy games
with arbitrary initial credit can be solved in pseudo-polynomial time
for fixed dimensions~3 or larger (Chaloupka, 2013).  It also improves
the complexity of solving multi-dimensional energy games with given
initial credit from non-elementary (Br\'azdil, Jan\v{c}ar, and
Ku\v{c}era, 2010) to \EXP[2], thus establishing
their \EXP[2]-completeness.

  \keywords{Energy game, bounding game, first-cycle game, vector
    addition system with states}
\end{abstract}
\maketitle

\section{Introduction}
\label{sec-intro}
\subsubsection*{Multi-Dimensional Energy Games\nopunct} are played
turn-by-turn by two players on a finite \emph{multi-weighted} game
graph, whose edges are labelled with integer vectors modelling
discrete energy consumption and refuelling.  Player~1's objective is
to keep the accumulated energy non-negative in every component along
infinite plays.  This setting is relevant to the synthesis of
resource-sensitive controllers balancing the usage of various
resources like fuel, time, money, or items in stock, and finding
optimal trade-offs;
see \citep{brazdil10,fahrenberg11,brazdil12,juhl13} for some examples.
Maybe more importantly, energy games are the key ingredient in the
study of several related resource-conscious games, notably
multi-dimensional mean-payoff games~\citep{chatterjee10} and games
played on vector addition systems with states
(VASS)~\citep{brazdil10,abdulla13,courtois14}.

The main open problem about these games has been to pinpoint the
complexity of deciding whether Player~1 has a winning strategy when
starting from a particular vertex and given an initial energy vector
as part of the input.  This particular \emph{given initial credit}
variant of energy games is also known as \emph{Z-reachability} VASS
games~\citep{brazdil10,chaloupka13}.  The problem is also equivalent
via logarithmic-space reductions to deciding \emph{single-sided} VASS
games with a non-termination objective~\citep{abdulla13}, and to
deciding whether a given VASS (or, equivalently, a Petri net)
simulates a given finite state
system~\ifshort\else\citep{jancar95,lasota09,courtois14,abdulla14}\fi.
As shown by \citet*{brazdil10}, all these problems can be solved in
$\EXP[(d-1)]$ where $d\geq 2$ is the number of energy components,
i.e.\ a \TOWER\ of exponentials when $d$ is part of the input.  The
best known lower bound for this problem is \EXP[2]-hardness
~\cite{courtois14},
leaving a substantial complexity gap.  So far, the only tight
complexity bounds are for $d=2$: \citet{chaloupka13} shows
the problem to be \P-complete when using unit updates, i.e.\ when the
energy levels can only vary by $-1$, $0$, or $1$.  However,
quoting \citeauthor{chaloupka13}, `since the presented results about
2-dimensional VASS are relatively complicated, we suspect this
[general] problem is difficult.'

When inspecting the upper bound proof of \citet{brazdil10}, it turns
out that the main obstacle to closing the gap and proving
\EXP[2]-completeness lies in the complexity upper bounds for energy
games with an \emph{arbitrary initial credit}---which is actually the
variant commonly assumed when talking about energy games.  Given a
multi-weighted game graph and an initial vertex $v$, we now wish to
decide whether there exists an initial energy vector $\vec b$ such
that Player~1 has a winning strategy starting from the pair $(v,\vec
b)$.  As shown by \citet*{chatterjee10}, this variant is simpler: it
is \coNP-complete.  However, the parameterised complexity bounds in
the literature~\citep{brazdil10,chatterjee14} for this simpler problem
involve an exponential dependency on the number $|V|$ of vertices in
the input game graph, which translates into a tower of exponentials
when solving the given initial credit variant.

\subsubsection*{Contributions.}
We show in this paper that the arbitrary initial credit problem for
$d$-dimensional energy games can be solved in time $O(|V| \cdot
\|E\|)^{O(d^4)}$ where $|V|$ is the number of vertices of the
input multi-weighted game graph and $\|E\|$ the maximal value
that labels its edges, and also deduce that the given initial credit
problem is solvable in time $O(|V| \cdot \|E\|)^{2^{O(d \cdot \log d)}}$
(see \autoref{cor:energy.upper}).  Both bounds are pseudo-polynomial 
when the dimension is fixed, and the latter establishes \EXP[2]-completeness 
closing the gap left open in~\citep{brazdil10,courtois14}.
Our parameterised bounds are of practical interest because typical
instances of energy games would have small dimension but might have a
large number of vertices.
\ifshort\else\par\fi
By the results of \citet{chatterjee10}, another consequence is that we
can decide the existence of a \emph{finite-memory} winning strategy
for fixed-dimensional \emph{mean-payoff} games in pseudo-polynomial
time.  The existence of a finite-memory winning strategy is the most
relevant problem for controller synthesis, but until now, solving
fixed-dimensional mean-payoff games in pseudo-polynomial time required
infinite memory strategies~\citep{chatterjee13}.

\subsubsection*{Overview.}  We prove our upper bounds on the
complexity of the arbitrary initial credit problem for $d$-dimensional
energy games by reducing them to \emph{bounding games}, where Player~1
additionally seeks to prevent arbitrarily high energy levels
(\autoref{sub:bounding}).  We further show these games to be
equivalent to \emph{first-cycle bounding games}
in \autoref{sec-fcycle}, where the total effect of the first simple
cycle defined by the two players determines the winner.  More
precisely, first-cycle bounding games rely on a hierarchically-defined
colouring of the game graph by \emph{perfect half-spaces}
(see \autoref{sec-pphs}), and the two players strive respectively to
avoid or produce cycles in those perfect half-spaces.

First-cycle bounding games coloured with perfect half-spaces can be
seen as generalising quite significantly both
\ifshort\relax\else\begin{itemize}\item\fi the `local strategy'
approach of \citet{chaloupka13} for 2-dimensional energy games, and
\ifshort\else\item\fi the `separating hyperplane technique'
of \citet{chatterjee13} for multi-dimensional mean-payoff
games\ifshort\else; see \appref{sub:meanpayoff} for an overview of the latter approach.\fi
\ifshort\relax\else\end{itemize}\fi

The reduction to first-cycle bounding games has several important
corollaries: the \emph{determinacy} of bounding games, and the
existence of a \emph{small hypercube property},
which in turn allow to derive the announced
complexity bounds on energy games (see \autoref{sec-cmplx}).  In fact,
we found with first-cycle bounding games a highly versatile tool, which
we use extensively in our proofs on energy games.

We start by presenting the necessary background on energy and bounding
games in \autoref{sec-wgames}.  Some omitted material \ifshort
\else 
on linear algebra can be found in \autoref{sec-algebra}.\fi

\section{Multi-Weighted Games}
\label{sec-wgames}
We define in this section the various games we consider in this work.
We start by defining multi-weighted game graphs, which provide a
finite representation for the infinite arenas over which our games are
played.  We then define energy games in \autoref{sub:energy}, and
their generalisation as bounding games in \autoref{sub:bounding}.

\subsection{Multi-Weighted Game Graphs}
\label{sub:wgames}
We consider game graphs whose edges are labelled by vectors of integers.
They are tuples of the form $\tuple{V, E, d}$, where $d$ is the
dimension in $\+N$,
$V\eqdef V_1\uplus V_2$ is a finite set of vertices, which is
  partitioned into Player~$1$ vertices ($V_1$) and Player~$2$ vertices
  ($V_2$), and 
$E$ is a finite set of edges included in $V \times\+Z^d\times V$,
and such that every vertex has at least one outgoing edge; we call the
{\ifshort\else edge\fi} labels in $\mathbb{Z}^d$ `weights'.

\begin{example}\label{ex-wgame}
  Figure~\ref{fig-wgame} shows an example of a
  2-dimensional multi-weighted game graph on its left-hand-side.  Throughout this paper,
  Player~1 vertices are depicted as triangles and Player~2 vertices as
  squares.
\end{example}

\begin{figure}[tbp]
  \centering
  \begin{tikzpicture}[auto,on grid,node distance=3cm]
  \node[triangle](C){$v_0$};
  \node[square,left=of C](L){$v_L$};
  \node[square,right=of C](R){$v_R$};
  \path[->,every node/.style={font=\footnotesize}]
  (C) edge[swap] node{$\tuple{0,0}$} (L)
  (C) edge node{$\tuple{0,0}$} (R)
  (L) edge[bend left=40,draw=black!40!red] node{$\tuple{-2,2}$} (C)
  (L) edge[bend right=26,draw=black!40!orange,swap] node{$\tuple{-1,3}$} (C)
  (R) edge[bend right=40,draw=black!40!blue,swap] node{$\tuple{2,-1}$} (C)
  (R) edge[bend left=26,draw=black!40!cyan] node{$\tuple{3,-3}$} (C);
  \end{tikzpicture}\hspace*{2cm}
  \begin{tikzpicture}
    \draw[step=.25cm,gray!40,very thin] (-.9,-.9) grid (.9,.9);
    \draw (-.9,0) -- (.9,0);
    \draw (0,-.9) -- (0,.9);
    \draw[very thick,black!40!red,->] (0,0) -- (-.5cm,.5cm);
    \draw[very thick,black!40!orange,->] (0,0) -- (-.25cm,.75cm);
    \draw[very thick,black!40!blue,->] (0,0) -- (.5cm,-.25cm);
    \draw[very thick,black!40!cyan,->] (0,0) -- (.75cm,-.75cm);
  \end{tikzpicture}
  \caption{A 2-dimensional multi-weighted game
  graph.\label{fig-wgame}}
\end{figure}
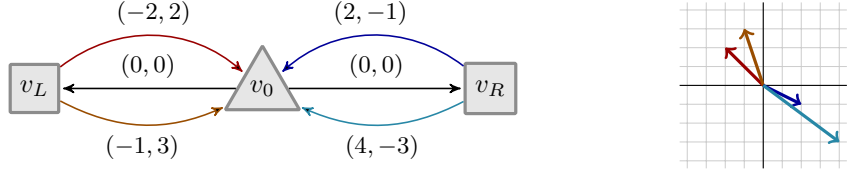

\subsubsection{Norms.}
For a vector $\vec{a}$, we denote the maximum absolute value of its
entries by $\|\vec{a}\|\eqdef\max_{1\leq i\leq d}|\vec a(i)|$, and we
call it the \emph{norm} of~$\vec{a}$.  By extension, for a set of
edges $E$, we let $\|E\|\eqdef\max_{(v,\vec u,v')\in E}\|\vec u\|$.
We assume, without loss of generality, that $\|E\| > 0$ in our
multi-weighted game graphs.  Regarding complexity, we encode
vectors of integers in binary, hence $\|E\|$ may be exponential in the
size of the multi-weighted game graph.

\subsubsection{Paths and Cycles.}
Given a multi-weighted game graph $\tuple{V, E, d}$, a
\emph{configuration} is a pair $(v,\vec a)$ with $v$ in $V$ and
$\vec a$ in $\+Z^d$.  A \emph{path} is a finite
sequence of configurations $\pi=(v_0,\vec a_0)(v_1,\vec
a_1)\cdots(v_n,\vec a_n)$ in $(V\times\+Z^d)^\ast$ such that for every
$0\leq j<n$ there exists an edge $(v_{j},\vec a_{j+1}-\vec
a_j,v_{j+1})$ in $E$ (where addition is performed componentwise).
The \emph{total weight} of such a path $\pi$ is
$w(\pi)\eqdef\sum_{0\leq j<n}\vec a_{j+1}-\vec a_j=\vec a_n-\vec a_0$.
\ifshort\else\par\fi
A \emph{cycle} is a path $(v_0,\vec a_0)(v_1,\vec a_1)\cdots(v_n,\vec
a_n)$ with $v_0=v_n$.  Such a cycle is \emph{simple} if 
$v_j=v_{k}$ for some $0\leq j<k\leq n$ implies $j=0$ and $k=n$.  We
assume, without loss of generality, that every cycle contains at least
one Player~$1$ vertex.  We often identify simple cycles with their
respective weights; the weights of the four simple cycles of the game
graph in~\autoref{fig-wgame} are displayed on its right-hand-side.
\begin{proposition}
\label{pr:ssw}
In any game graph $\tuple{V, E, d}$,
the total weight of any simple cycle has norm at most
$|V| \cdot \|E\|$.
\end{proposition}

\subsubsection{Plays and Strategies.} 
Let $v_0$ be a vertex from $V$.  A \emph{play from $v_0$} is an
infinite configuration sequence $\rho=(v_0,\vec a_0)(v_1,\vec
a_1)\cdots$ such that $\vec a_0=\vec 0$ is the null vector and every
finite prefix $\rho|_n\eqdef(v_0,\vec a_0)\cdots(v_n,\vec a_n)$ is a
path.  Note that, because $\vec a_0=\vec 0$, the total weight of this
prefix is $w(\rho|_n)=\vec a_n$.  We define the \emph{norm} of a play
$\rho$ as the supremum of the norms of total weights of its prefixes:
$\|\rho\|\eqdef\sup_n\|w(\rho|_n)\|$.
\ifshort\else\par\fi
A \emph{strategy} for Player~$p$, $p\in\{1,2\}$, is a function
$\sigma_p$ taking as input a non-empty path $\pi\cdot(v,\vec a)$
ending in a Player~$p$ vertex $v\in V_p$, and returning an edge
$\sigma_p(\pi\cdot(v,\vec a))=(v,\vec u,v')$ from $E$.
\ifshort
\else A play $\rho=(v_0,\vec a_0)(v_1,\vec a_1)\cdots$ is \emph{consistent}
with a strategy $\sigma_p$ for Player~$p$ if whenever $v_n$ is a
Player~$p$ vertex in $V_p$, then $\sigma_p(\rho|_n)=(v_n,\vec
a_{n+1}-\vec a_n,v_{n+1})$.
Given strategies $\sigma_1$ and $\sigma_2$ for Player~1 and
Player~2 respectively, and an initial vertex $v_0$, observe
that there is a unique play $\rho_{v_0,\sigma_1,\sigma_2}$ from $v_0$
consistent with both $\sigma_1$ and~$\sigma_2$.\fi

\begin{example}[continues=ex-wgame]
For instance, in the game graph depicted in \autoref{fig-wgame}, a
strategy for Player~1 could be to move to $v_L$ whenever the current
energy level on the first coordinate is non-negative, and to $v_R$
otherwise---note that this is an \emph{infinite-memory} strategy---:
\begin{align}\label{eq-strategy1}
  \sigma_1(\pi\cdot(v_0,\vec
  a))&\eqdef\begin{cases}(v_0,\tuple{0,0},v_L)&\text{if $\vec a(1)\geq
  0$,}\\
  (v_0,\tuple{0,0},v_R)&\text{otherwise,}\end{cases} \intertext{and
  one for Player~2 could be to always select one particular edge in
  every vertex, regardless of the current energy vector---this is
  called a \emph{counterless}
  strategy~\citep{brazdil10}---:}\label{eq-strategy2} \sigma_2(\pi\cdot(v,\vec
  a))&\eqdef\begin{cases}(v_L,\tuple{-2,2},v_0)&\text{if $v=v_L$}\\
  (v_R,\tuple{2,-1},v_0)&\text{otherwise.}\end{cases}
\end{align}\ifshort\else
  These strategies define a unique consistent play for $v_0$, which starts with
\begin{equation}
  (v_0,0,0)(v_L,0,0)(v_0,-2,2)(v_R,-2,2)(v_0,0,1)(v_L,0,1)(v_0,-2,3)\cdots\qedhere
\end{equation}\fi
\end{example}
\ifshort\else
In the following we consider several different winning conditions on
plays, which define different games played on multi-weighted game
graphs.\fi

\subsection{Multi-Dimensional Energy Games}
\label{sub:energy}
Suppose $\tuple{V, E, d}$ is a multi-weighted game graph, $v_0$ an
initial vertex, and $\vec{b}$ is a vector from $\+N^d$.  A play $\rho$
from $v_0$ is \emph{winning} for Player~1 in
the \emph{energy game} $\Delta_{\vec{b}}\!\tuple{V, E, d}$
with \emph{initial credit} $\vec b$ if, for all $n$, $\vec
b+w(\rho|_n)\geq\vec 0$, using the product ordering over $\+Z^d$.
Otherwise, Player~2 wins the play.
\ifshort\else As usual, this means that Player~1 wins the energy game $\Delta_{\vec
b}\!\tuple{V,E,d}$ from $v_0$ if there exists a \emph{winning
strategy} $\sigma_1$ for Player~1, i.e.\ $\sigma_1$ is such that for
all strategies $\sigma_2$ for Player~2 the play
$\rho_{v_0,\sigma_1,\sigma_2}$ is winning for Player~1.  \fi An
immediate property of energy games is \emph{monotonicity}: if
$\sigma_1$ is winning for Player~1 with some initial credit $\vec b$,
and $\vec b'\geq\vec b$, then it is also winning for Player~1 with
initial credit~$\vec b'$.
\begin{example}[continues=ex-wgame]
For example, one may observe that the strategy~\eqref{eq-strategy1}
for Player~1 is winning for the game graph of \autoref{fig-wgame} with
initial credit $\tuple{2,2}$ (or larger).
\ifshort\else\par\fi
A geometric intuition comes from the directions of the total weights
of simple cycles in \autoref{fig-wgame}: by choosing alternatively
edges to $v_L$ or $v_R$, Player~1 is able to balance the energy levels
above the `$x+y=0$' line.\ifshort\else\ One way to see this more
formally is to build the corresponding \emph{self-covering strategy
tree} up to the first time when a configuration is greater or equal to
another configuration higher in the tree~\citep{brazdil10}.  By
monotonicity of the game, Player~1 can repeat the same actions
from those leaves.  See \autoref{fig-strat-tree} for our example.

Strategy $\sigma_1$ uses the comparison of $\vec a(1)$ with $0$ as
a \emph{soft bound} to trigger a change of strategy and attempt to
forbid cycles with a negative effect on the first coordinate.  Note
that the energy level $\vec a(1)$ might nevertheless become less than
$0$, but will remain $\geq -2$ at all times; we call this
the \emph{hard bound}.  This follows the general scheme
of \citet{chaloupka13}---and also ours---for Player~1 strategies.
\begin{figure}[tbp]
  \centering
  \begin{tikzpicture}[on grid,every
    node/.style={font=\footnotesize,inner sep=1pt},level distance=.8cm,thick]
    \node(r){$\tuple{v_0,0,0}$}
    child{node{$\tuple{v_L,0,0}$}
      child[sibling distance=3.2cm,draw=black!40!red]{node(n){$\tuple{v_0,-2,2}$}
        child[draw=black]{node{$\tuple{v_R,-2,2}$}
          child[sibling distance=1.6cm,draw=black!40!blue]{node(l1){$\tuple{v_0,0,1}$}}
          child[sibling distance=1.6cm,draw=black!40!cyan]{node{$\tuple{v_0,1,-1}$}
            child[draw=black]{node{$\tuple{v_L,1,-1}$}
              child[draw=black!40!red]{node{$\tuple{v_0,-1,1}$}
                child[draw=black]{node{$\tuple{v_R,-1,1}$}
                  child[sibling distance=1.6cm,draw=black!40!blue]{node(l2){$\tuple{v_0,1,0}$}}
                  child[sibling distance=1.6cm,draw=black!40!cyan]{node{$\tuple{v_0,2,-2}$}
                    child[draw=black]{node{$\tuple{v_L,2,-2}$}
                      child[sibling distance=1.6cm,draw=black!40!red]{node(l6){$\tuple{v_0,0,0}$}}
                      child[sibling distance=1.6cm,draw=black!40!orange]{node(l7){$\tuple{v_0,1,1}$}}}}}}
              child[draw=black!40!orange]{node(l3){$\tuple{v_0,0,2}$}}}}}}
      child[sibling distance=3.2cm,draw=black!40!orange]{node{$\tuple{v_0,-1,3}$}
        child[draw=black]{node{$\tuple{v_R,-1,3}$}
          child[sibling distance=1.6cm,draw=black!40!blue]{node(l4){$\tuple{v_0,1,2}$}}
          child[sibling distance=1.6cm,draw=black!40!cyan]{node(l5){$\tuple{v_0,2,0}$}}}}};
    \path[dotted,->,color=gray,every node/.style={font=\tiny,inner sep=.2pt},thin]
      (l1) edge[bend left=50] node[fill=white]{$\leq$} (r)
      (l2.west) edge[bend left=80] node[fill=white]{$\leq$} (r.west)
      (l3) edge[bend right=28] node[fill=white]{$\leq$} (n.east)
      (l4) edge node[fill=white]{$\leq$} (r)
      (l5) edge[bend right=50] node[fill=white]{$\leq$} (r)
      (l6) edge[bend left=100] node[fill=white]{$\leq$} (r)
      (l7) edge[bend right=100] node[fill=white]{$\leq$} (r);           
  \end{tikzpicture}
  \caption{Self-covering strategy tree for Player~1 in the energy
  game of \autoref{fig-wgame}.\label{fig-strat-tree}}
\end{figure}
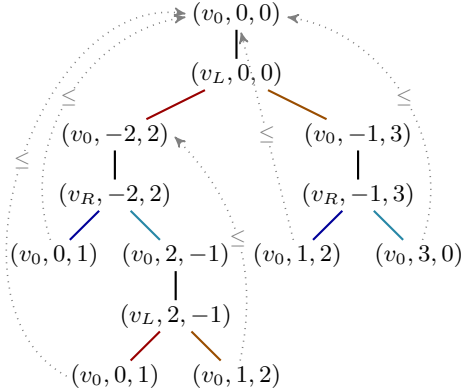\fi
\end{example}

\subsection{Multi-Dimensional Bounding Games}
\label{sub:bounding}
A generalisation of energy games sometimes considered in the
literature is to further impose a maximal \emph{capacity} $\vec
c\in\+N^d$ (also called an upper bound) on the energy levels during
the play~\cite{fahrenberg11,juhl13}.  Player~1 then wins a
play $\rho$ if $0\leq\vec b+w(\rho|_n)\leq\vec c$ for all $n$.

In the spirit of the arbitrary initial credit variant of energy games,
we also quantify $\vec c$ existentially.  This defines the
\emph{bounding game} $\Gamma\!\tuple{V, E, d}$ over a multi-weighted
game graph $\tuple{V, E, d}$, where a play $\rho$ is winning for
Player~1 if its norm $\|\rho\|$ is finite, i.e.\ if the set
$\{\|w(\rho {\mid_n})\| \,:\, n \in \mathbb{N}\}$ of norms of total
weights of all finite prefixes of $\rho$ is bounded, and Player~2 wins
otherwise, i.e.\ if the set is unbounded.  In other words, Player~1 strives to
contain the current vector within some $d$-dimensional hypercube,
while Player~2 attempts to escape.

\begin{example}[continues=ex-wgame]
Note that Player~2 is now winning the bounding game defined by the
game graph of \autoref{fig-wgame} from any of the three vertices, for
example using the strategy~\eqref{eq-strategy2}.  Indeed, this strategy
ensures that the only simple cycles that can be played have
weights~$(-2,2)$ and~$(2,-1)$.  Because these vectors belong to an open
half-plane, the total energy will drift deeper and deeper inside that
open half-plane and its norm will grow unbounded.
\end{example}

\begin{figure}[tbp]
  \centering
  \begin{tikzpicture}[auto,on grid]
    \node[triangle](L){$v_L$};
    \node[triangle, right=3cm of L](R){$v_R$};
    \path[->,every node/.style={font=\footnotesize}]
      (L) edge[bend left=25,draw=black!40!cyan] node{$\tuple{-1,0}$} (R)
      (R) edge[bend left=25,draw=black!40!cyan] node{$\tuple{0,-1}$} (L)
      (L) edge[loop above,draw=black!40!blue] node{$\tuple{1,-1}$} ()
      (R) edge[loop above,draw=black!40!violet] node{$\tuple{-1,1}$} ();
  \end{tikzpicture}\hspace*{2cm}
  \begin{tikzpicture}
    \draw[step=.5cm,gray!40,very thin] (-.6,-.6) grid (.6,.6);
    \draw (-.6,0) -- (.6,0);
    \draw (0,-.6) -- (0,.6);
    \draw[very thick,black!40!violet,->] (0,0) -- (-.5cm,.5cm);
    \draw[very thick,black!40!cyan,->] (0,0) -- (-.5cm,-.5cm);
    \draw[very thick,black!40!blue,->] (0,0) -- (.5cm,-.5cm);
    \node at (0,-1.2) {};
  \end{tikzpicture}\vspace*{-.5em}
  \caption{A 2-dimensional game graph with only Player~1
    vertices.\label{fig:mean}}
\end{figure}
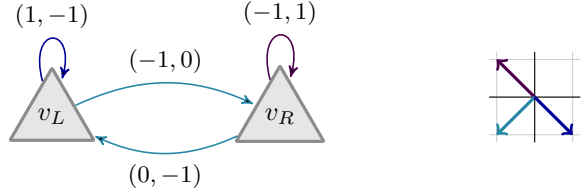

\begin{example}\label{ex:mean}
As a rather different example, consider the multi-weighted game graph
of \autoref{fig:mean}.  Although Player~2 does not control any vertex,
and Player~1 controls the `direction of divergence', Player~2
wins the associated bounding game.
Indeed, Player~1 can either eventually stay forever at one of the two
vertices, or visit both vertices infinitely often.  Anyway, she
loses.
\end{example}

\section{Complexity Upper Bounds}\label{sec-cmplx}

Our main results are new parameterised complexity upper bounds for
deciding whether Player~1 has a winning strategy in a given energy
game.  In turn, we rely for these results on a \emph{small hypercube
property} of bounding games, which we introduce next, and which will
be a consequence of the study of first-cycle bounding games
in \autoref{sec-fcycle}.

\subsection{Small Hypercube Property}
In a bounding game, if Player~1 is winning, then by definition she has
a winning strategy $\sigma_1$ such that for all plays $\rho$
consistent with $\sigma_1$ there exists some bound $B_\rho$ with
$\|\rho\|\leq B_\rho$.  We considerably strengthen this statement
in \autoref{sec-fcycle} where we construct an explicit winning
strategy, which yields an explicit \emph{uniform} bound~$B$ for all
consistent plays:
\begin{lemma}\label{lem-shp}
  Let $\tuple{V,E,d}$ be a multi-weighted game graph.  If Player~1
  wins the bounding game $\Gamma\!\tuple{V,E,d}$, then she has a
  winning strategy which ensures \ifshort$\else\begin{equation*}\fi\|\rho\|\leq (4 |V| \cdot \|E\|)^{2
  (d + 2)^3}\ifshort$\else\end{equation*}\fi for all consistent plays $\rho$.
\end{lemma}
Note that our bound is polynomial in $|V|$ the number of vertices,
unlike the bounds found in comparable statements 
by \citet[\lemmaautorefname~7]{brazdil10}
and \citet[\lemmaautorefname~3]{chatterjee14}, which incur
an exponential dependence on $|V|$.  This entails pseudo-polynomial 
complexity bounds when $d$ is fixed:

\begin{corollary}
\label{cor:haupt}
  Bounding games on multi-weighted graphs $\tuple{V, E, d}$ are
  solvable in \ifshort\else deterministic \fi time
  $O(|V| \cdot \|E\|)^{O(d^4)}$.
\end{corollary}
\begin{proof}
By \autoref{lem-shp}, the bounding game is equivalent to a
reachability game where Player~2 attempts to see the norm of the total
weight exceed $B\eqdef(4 |V| \cdot \|E\|)^{2 (d + 2)^3}$.  This can be
played within a finite arena of size $(2B+1)^d$ and solved in time linear
in that size using the usual attractor computation algorithm.
\end{proof}

\subsection{Solving Energy Games}
We now show how energy games can be solved by
solving bounding games on appropriately augmented game graphs.
Given an energy game and an initial vertex, there are two standard decision problems:
\begin{description}
\item[arbitrary initial credit]
does Player~1 win for some initial energy vector;
\item[given initial credit]
does Player~1 win for a given initial energy vector?
\end{description}
Another problem of interest is computing the Pareto limit, i.e. the set of all pointwise minimal initial energy vectors for which Player~1 wins.

It will turn out that solving the arbitrary initial credit problem is relatively easy
since it amounts to solving the bounding game with self-loops added at all Player~1 vertices
that give her power to prevent Player~2 from winning by diverging in a non-negative direction.
Solving the given initial credit problem and computing the Pareto limit will require more work,
involving reasoning \`a la \citeauthor{rackoff78}'s for the covering problem,
where `very small' vector components which may potentially become negative
are tracked using graph vertices and
the remaining `very large' vector components are guaranteed to remain so
by the small hypercube property of bounding games.
Although the presentation of the latter work takes up most of this section,
we remark that it is relatively uninteresting and essentially follows
the pattern already seen in \citeauthor{brazdil10}'s stepping up
from the arbitrary to the given initial credit problem.

\ifshort\else\subsubsection{Tracking Sets.}
Supposing $\tuple{V, E, d}$ is a multi-weighted game graph,
let $\lambda_i = (4 |V| \cdot \|E\|)^{(3 d)^{4 i}}$ for $i = 1, \ldots, d$.

We write $\+N_\infty$ for the natural numbers extended by $\infty$,
where $n < \infty$, $\infty + n = \infty$, and $\infty - n = \infty$, for all $n \in \+N$.

Inspired by \citet[Section~8]{leroux13}, we say that $T \subseteq \{1, \ldots, d\}$
is a \emph{tracking set} for $\vec{a} \in \+N_\infty^d$ iff
$\vec{a}(i) \geq \lambda_{|T| + 1}$ for all $i \,\in\, \{1, \ldots, d\} \setminus T$.
(We remark that if $|T| = d$ then the quantification over $i$ is empty,
so it is not an issue that $\lambda_{|T| + 1}$ is undefined in that case.)

Suppose $\vec{a} \in \+N_\infty^d$.
Observing that $\{1, \ldots, d\}$ is a tracking set,
and that the class of tracking sets is closed under intersection
since the sequence $\lambda_d, \ldots, \lambda_1$ is decreasing,
we conclude that $\vec{a}$ has a \emph{unique minimal} tracking set~$\mathsf{tr}(\vec{a})$.
We write $\lambda(\vec{a})$ for the abstraction $\vec{a}|_{\mathsf{tr}(\vec{a})} \in \+N_\infty^d$
obtained from $\vec{a}$ by setting to $\infty$ all its components
whose indices are outside of its unique minimal tracking set.

Intuitively, for every $\vec{a} \in \+N_\infty^d$,
its unique minimal tracking set $\mathsf{tr}(\vec{a})$ identifies
the least collection of components of $\vec{a}$ that are much smaller than
the remaining components of~$\vec{a}$,
in the sense that they are separated by the two corresponding consecutive elements of
the increasing sequence $\lambda_1, \ldots, \lambda_d$.
Indeed, we have that $\vec{a}(i) \geq \lambda_{|\mathsf{tr}(\vec{a})| + 1}$
for all $i \not\in \mathsf{tr}(\vec{a})$ since $T(\vec{a})$ is a tracking set,
and that $\vec{a}(i) < \lambda_{|\mathsf{tr}(\vec{a})|}$
for all $i \in \mathsf{tr}(\vec{a})$ since $\mathsf{tr}(\vec{a})$ is minimal.
(We remark that if $|\mathsf{tr}(\vec{a})| = 0$ then the latter quantification over $i$ is empty,
so it is not an issue that $\lambda_{|\mathsf{tr}(\vec{a})|}$ is undefined in that case.)

We have called those sets `tracking' because,
in some constructions to follow shortly,
we shall use them to identify vector components
whose values will be exactly tracked in vertices of game graphs.

For $T \in \{1, \ldots, d\}$, let $\+N_\infty^d|_T$ consist of all
$\vec{a} \in \+N_\infty^d$ such that $\mathsf{tr}(\vec{a}) = T$ and $\lambda(\vec{a}) = \vec{a}$.
Informally, it is the collection of all vectors whose minimal set of `very small' components is $T$
and which have been abstracted by setting all their other components to~$\infty$.
By what we observed above, the cardinality of $\+N_\infty^d|_T$ is at most
$\left(\lambda_{|T|}\right)^{|T|}$.
(We remark that if $T$ is the empty set then this is~$1$.)

\subsubsection{Tracking Lossy Game Graphs.}
We now define $\mathsf{TrLo}_T\!\tuple{V, E, d}$ as game graphs
which are obtained from the game graph $\tuple{V, E, d}$
by extending the vertices so that they can exactly track
vector components at indices in the set $T \in \{1, \ldots, d\}$ as long as
they are `much smaller' than the remaining ones,
and by inserting `lossy' self-loops that enable Player~1 to prevent
the latter vector components from becoming unboundedly large.
These constructions will give us the means to reduce solving energy games to solving bounding games:
intuitively, since bounding games have the small hypercube property (cf.\ \autoref{lem-shp}),
Player~1 in an energy game can treat `sufficiently large' vector components as if
she was playing a bounding game, but she has to be careful with keeping the others non-negative;
and she should be given power artificially for preventing Player~2 from
winning the bounding game by making values unboundedly large
since that would not constitute a win for him in the energy game.

The definition is recursive, where we assume that $\mathsf{TrLo}_T\!\tuple{V, E, d}$ game graphs
have been defined for all $T$ of smaller cardinality (if any).
We write $\overline{T}$ for $\{1, \ldots, d\} \setminus T$.
\begin{itemize}

\item
The dimension is $\left|\overline{T}\right| = d - |T|$,
where we regard edge weights as vectors in~$\+Z^{\overline{T}}$.

\item
The set of Player~1 vertices is $\left(V_1 \times \+N_\infty^d|_T\right) \uplus \{v_\top, v_\bot\}$,
and the set of Player~2 vertices is $V_2 \times \+N_\infty^d|_T$.

\item
For every $\tuple{v, \vec{u}, v'} \in E$ and $\vec{a} \in \+N_\infty^d|_T$, there are four cases:
\begin{description}[style=unboxed]

\item[$\lambda(\vec{a} + \vec{u})$ is in $\+N_\infty^d|_T$]
there is an edge $\tuple{\tuple{v, \vec{a}}, \vec{u}|_{\overline{T}}, \tuple{v', \vec{a} + \vec{u}}}$;

\item[$\lambda(\vec{a} + \vec{u})$ is in $\+N_\infty^d|_{T'}$ for a strict subset $T'$ of $T$ and
      Player~1 wins $\Gamma\!\left(\mathsf{TrLo}_{T'}\!\tuple{V, E, d}\right)$
      from $\tuple{v', \lambda(\vec{a} + \vec{u})}$]
there is an edge $\tuple{\tuple{v, \vec{a}}, \vec{u}|_{\overline{T}}, v_\top}$;

\item[$\lambda(\vec{a} + \vec{u})$ is in $\+N_\infty^d|_{T'}$ for a strict subset $T'$ of $T$ and
      Player~1 loses $\Gamma\!\left(\mathsf{TrLo}_{T'}\!\tuple{V, E, d}\right)$
      from $\tuple{v', \lambda(\vec{a} + \vec{u})}$]
there is an edge $\tuple{\tuple{v, \vec{a}}, \vec{u}|_{\overline{T}}, v_\bot}$;

\item[$\vec{a} + \vec{u}$ has a negative component]
there is an edge $\tuple{\tuple{v, \vec{a}}, \vec{u}|_{\overline{T}}, v_\bot}$.
\end{description}

\item
Every Player~1 vertex of $\mathsf{TrLo}_T\!\tuple{V, E, d}$, for every $i \in \overline{T}$,
has a self-loop labelled by the negative unit vector $-\vec{e}_i$.

\item
The only other edge of $\mathsf{TrLo}_T\!\tuple{V, E, d}$ is a self-loop on vertex $v_\top$
labelled by the $\left|\overline{T}\right|$-dimensional zero vector.
\end{itemize}

The conditions that define the four cases above have the property that
always exactly one of them is satisfied because
the $\Gamma\!\left(\mathsf{TrLo}_{T'}\!\tuple{V, E, d}\right)$ bounding game is determined,
cf.\ \autoref{sec-fcycle}.

When $T = \{1, \ldots, d\}$, according to the definition above,
$\mathsf{TrLo}_T\!\tuple{V, E, d}$ has dimension~$0$,
so it fails the assumption from\autoref{sub:wgames} that
the maximum norm of its edge weights is positive.
In that case also, its vertex $v_\bot$ fails the assumption from the same section that
every vertex has at least one outgoing edge.
Nevertheless, $\mathsf{TrLo}_{\{1, \ldots, d\}}\!\tuple{V, E, d}$ is otherwise well defined,
and we shall regard $\Gamma\!\left(\mathsf{TrLo}_{\{1, \ldots, d\}}\!\tuple{V, E, d}\right)$ as
a reachability game in which the goal of Player~1 is to avoid the vertex $v_\bot$ forever
and the goal of Player~2 is to reach it.
(We remark that in this game graph the vertex $v_\top$ has a self-loop
and is hence winning for Player~1 in the reachability game.)

At the other extreme, when $T$ is empty, $\+N_\infty^d|_T$ is the singleton set
consisting of the vector with $\infty$ in every component.
Moreover, the vertices $v_\top$ and $v_\bot$ are not reachable from other vertices,
and thus can be removed from the game graph.
Therefore, $\mathsf{TrLo}_\emptyset\!\tuple{V, E, d}$ is essentially
$\tuple{V, E, d}$ extended by the lossiness,
i.e.\ the self-loops at every Player~1 vertex and with every negative unit weight.

\subsubsection{Solving Energy Games Using Bounding Games.}
We now show that solving an energy game is reducible to
solving a bounding game on a tracking and lossy game graph as defined above,
where the set of `very small' vector components to track is determined by the initial credit vector.

\begin{theorem}
\label{th:energy.bounding}
For every $p \in \{1, 2\}$, initial credit $\vec{b} \in \+N^d$, and vertex~$v$, we have that
Player~$p$ wins energy game $\Delta_{\vec{b}}\!\tuple{V, E, d}$ from $v$ if Player~$p$ wins
bounding game $\Gamma\!\left(\mathsf{TrLo}_{\mathsf{tr}(\vec{b})}\!\tuple{V, E, d}\right)$
from $\tuple{v, \lambda(\vec{b})}$.
\end{theorem}

\begin{proof}
The proof is inductive, where the hypothesis is that the statement holds
for every smaller cardinality of $\mathsf{tr}(\vec{b})$.

We first handle Player~1,
so suppose she has a winning strategy $\sigma$
in $\Gamma\!\left(\mathsf{TrLo}_T\!\tuple{V, E, d}\right)$
from $\tuple{v, \lambda(\vec{b})}$,
where $T$ is the unique minimal tracking set of vector~$\vec{b}$.
Intuitively, we shall obtain a winning strategy for Player~1 in the energy game
by playing according to $\sigma$ as long as the vector components
that have not been abstracted away in the initial credit $\vec{b}$ remain `very small',
and by switching to a winning strategy provided from the inductive hypothesis
as soon as one or more of those vector components become `too large'.
The tracking of the former vector components in the bounding game
will ensure that they remain non-negative,
and the small hypercube property of bounding games
will ensure that the abstracted vector components remain sufficiently large.
All of the lossy self-loops that are performed by $\sigma$ can be skipped
since doing so only results in larger values.

When $|T| < d$, the game is bounding of dimension $d - |T|$,
so by \autoref{lem-shp} we can assume that $\sigma$ stays within a small hypercube.
If also $|T| > 0$, this ensures that the norm of the total weight
of every prefix of every consistent play is at most
\begin{multline*}
\left(4 \left(|V| \cdot \left|\+N_\infty^d|_T\right| + 2\right) \cdot \|E\|\right)
^{2 (d - |T| + 2)^3} \leq \\
\left(4 \left(|V| \cdot \left(\lambda_{|T|}\right)^{|T|} + 2\right) \cdot \|E\|\right)
^{2 (d - |T| + 2)^3} < \\
\left(4 \left(|V| \cdot (4 |V| \cdot \|E\|)^{d (3 d)^{4 |T|}} + 2\right) \cdot \|E\|\right)
^{2 (d - |T| + 2)^3} < \\
\left(8 |V| \cdot \|E\| \cdot (4 |V| \cdot \|E\|)^{d (3 d)^{4 |T|}}\right)^{2 (d - |T| + 2)^3} \leq \\
\left((4 |V| \cdot \|E\|)^{d (3 d)^{4 |T|} + 1.5}\right)^{2 (d - |T| + 2)^3} < \\
\left((4 |V| \cdot \|E\|)^{1.5 d (3 d)^{4 |T|}}\right)^{2 (d - |T| + 2)^3} \leq \\
(4 |V| \cdot \|E\|)^{(3 d)^{4 |T| + 1} (2 d)^3} <
\lambda_{|T| + 1} - \lambda_{|T|}\;.
\end{multline*}
If $|T| = 0$, the bound simplifies to
\[(4 |V| \cdot \|E\|)^{2 (d + 2)^3} < (4 |V| \cdot \|E\|)^{(3 d)^4} = \lambda_1\;.\]

Let $\widehat{\sigma}$ be the following strategy of Player~1
in the energy game $\Delta_{\vec{b}}\!\tuple{V, E, d}$ from the vertex~$v$.
\begin{itemize}

\item
Strategy $\widehat{\sigma}$ makes the same choices as the strategy $\sigma$
as long as playing by the latter does not reach the vertex~$v_\top$,
except that it skips any self-loops with negative unit weights
that were added to Player~1 vertices in the definition of $\mathsf{TrLo}_T\!\tuple{V, E, d}$.
This is well defined since every edge in $\mathsf{TrLo}_T\!\tuple{V, E, d}$ of the form
$\tuple{\tuple{v', \vec{a}}, \vec{u}|_{\overline{T}}, \tuple{v'', \vec{a} + \vec{u}}}$
determines the corresponding edge $\tuple{v', \vec{u}, v''}$ in $\tuple{V, E, d}$,
and since $\sigma$ cannot keep choosing the lossy self-loops
consecutively forever because it is winning.
Note also that, for the same reason, playing by $\sigma$ cannot reach the vertex~$v_\bot$.

\item
Observe that the portion of $\widehat{\sigma}$ defined so far has the following property.
For every vertex $\tuple{v', \vec{a}}$ and total weight $\vec{w} \in \+Z^{\overline{T}}$
which are reached by playing according to $\sigma$
in $\mathsf{TrLo}_T\!\tuple{V, E, d}$ from $\tuple{v, \lambda(\vec{b})}$,
we have that $v'$ is the corresponding vertex reached by
playing according to $\widehat{\sigma}$ in $\tuple{V, E, d}$ from~$v$,
and that the corresponding energy vector $\vec{b}'$ satisfies:
\begin{itemize}
\item
$\vec{b}'|_T = \vec{a}|_T$;
\item
for every $i \in \overline{T}$, we have that $\vec{b}'(i) \geq \vec{b}(i) + \vec{w}(i)$,
where the difference between the two sides is exactly the number of $-\vec{e}_i$ added self-loops
that have been taken by~$\sigma$.
\end{itemize}
The latter inequality, together with the small hypercube property of $\sigma$ and
the fact that $\vec{b}(i) \geq \lambda_{|T| + 1}$ since $T = \mathsf{tr}(\vec{b})$,
implies that $\vec{b}'(i) \geq 0$.

\item
If and as soon as playing by $\sigma$ reaches the vertex~$v_\top$,
which can only be by and edge in $\mathsf{TrLo}_T\!\tuple{V, E, d}$ of the form
$\tuple{\tuple{v', \vec{a}}, \vec{u}|_{\overline{T}}, v_\top}$ where
there exists an edge $\tuple{v', \vec{u}, v''}$ in $\tuple{V, E, d}$ such that
$\lambda(\vec{a} + \vec{u})$ is in $\+N_\infty^d|_{T'}$ for a strict subset $T'$ of $T$ and
Player~1 wins $\Gamma\!\left(\mathsf{TrLo}_{T'}\!\tuple{V, E, d}\right)$
from $\tuple{v'', \lambda(\vec{a} + \vec{u})}$,
then $\widehat{\sigma}$ follows such an edge $\tuple{v', \vec{u}, v''}$
which results in some current energy vector $\vec{b}''$
and subsequently plays according to some winning strategy of Player~1
in the energy game $\Delta_{\vec{b}''}\!\tuple{V, E, d}$ from the vertex~$v''$.
Such a strategy exists by the inductive hypothesis
because $\lambda(\vec{b}'') = \lambda(\vec{a} + \vec{u})$,
which follows from the observations above and,
in case $|T| < d$, from the small hypercube property of $\sigma$
(with recalling that now also $|T| > 0$).
\end{itemize}

It remans to handle Player~2,
so suppose he has a winning strategy $\tau$
in $\Gamma\!\left(\mathsf{TrLo}_T\!\tuple{V, E, d}\right)$
from $\tuple{v, \lambda(\vec{b})}$,
where $T$ is the unique minimal tracking set of vector~$\vec{b}$.
Intuitively, we shall obtain a winning strategy for Player~2 in the energy game
by playing according to $\tau$ as long as the vector components
that have not been abstracted away in the initial credit $\vec{b}$
remain non-negative and `very small'.
If that continues forever,
Player~2 will win regardless of how large the other components of $\vec{b}$ are,
because the added lossy self-loops in the bounding game ensure that
$\tau$ makes the total weight diverge in some negative direction.
Otherwise, Player~2 will win either immediately or
by switching to a strategy provided from the inductive hypothesis.

When $|T| < d$, we have that \autoref{l:Exp} \autoref{lem-p1}
apply to the bounding game of dimension $d - |T|$
and its first-cyle variant $G\!\left(\mathsf{TrLo}_T\!\tuple{V, E, d}\right)$,
so we can assume that $\tau$ mimics a winning strategy of Player~2
in $G\!\left(\mathsf{TrLo}_T\!\tuple{V, E, d}\right)$
from $\tuple{v, \lambda(\vec{b})}$ as in \autoref{sub-player2}.
Since every Player~1 vertex was augmented with every negative unit self-loop,
the latter strategy chooses only perfect half-spaces that are disjoint from
the non-negative orthant $\cone{\vec{e}_i \,:\, i \in \overline{T}}$.
Hence, from the proof of \autoref{l:Exp},
for every infinite play $\rho$ consistent with~$\tau$,
letting $C_1, C_2, \ldots$ be its cycle decomposition,
there exist an open half-subspace $H$ and a positive integer $N$ such that:
\begin{itemize}
\item
$H$ is disjoint from the non-negative orthant;
\item
for each $n \geq N$, the total weight $w(C_n)$ belongs to
$\overline{H}$ the topological closure of~$H$;
\item
the set of all distances of $w(C_N) + \cdots + w(C_n)$ from the boundary of $H$ is unbounded.
\end{itemize}

Let $\widehat{\tau}$ be the following strategy of Player~2
in the energy game $\Delta_{\vec{b}}\!\tuple{V, E, d}$ from the vertex~$v$.
\begin{itemize}

\item
Strategy $\widehat{\tau}$ makes the same choices as the strategy $\tau$
as long as playing by the latter does not reach the vertex~$v_\bot$.
This is well defined since every edge in $\mathsf{TrLo}_T\!\tuple{V, E, d}$ of the form
$\tuple{\tuple{v', \vec{a}}, \vec{u}|_{\overline{T}}, \tuple{v'', \vec{a} + \vec{u}}}$
determines the corresponding edge $\tuple{v', \vec{u}, v''}$ in $\tuple{V, E, d}$.
Note that the lossy self-loops that were added
in the definition of $\mathsf{TrLo}_T\!\tuple{V, E, d}$
do not arise here since they are at Player~1 vertices.
Note also that, since $\tau$ is winning, playing by it cannot reach the vertex~$v_\top$.

\item
Observe that the portion of $\widehat{\tau}$ defined so far
can produce an infinite play $\widehat{\rho}$ only when $|T| < d$,
because otherwise $\tau$ would admit an infinite play that avoids $v_\bot$ forever
and would thus not be winning in the reachability game
$\Gamma\!\left(\mathsf{TrLo}_T\!\tuple{V, E, d}\right)$.
Every such $\widehat{\rho}$ is then winning for Player~2 by the analysis above,
regardless of how large the initial credit vector components $\vec{b}|_{\overline{T}}$ are.

\item
One way in which playing by $\tau$ can reach the vertex~$v_\bot$
is by and edge in $\mathsf{TrLo}_T\!\tuple{V, E, d}$ of the form
$\tuple{\tuple{v', \vec{a}}, \vec{u}|_{\overline{T}}, v_\bot}$ where
there exists an edge $\tuple{v', \vec{u}, v''}$ in $\tuple{V, E, d}$ such that
$\lambda(\vec{a} + \vec{u})$ is in $\+N_\infty^d|_{T'}$ for a strict subset $T'$ of $T$ and
Player~2 wins $\Gamma\!\left(\mathsf{TrLo}_{T'}\!\tuple{V, E, d}\right)$
from $\tuple{v'', \lambda(\vec{a} + \vec{u})}$.
Then $\widehat{\tau}$ follows such an edge $\tuple{v', \vec{u}, v''}$
which results in some current energy vector $\vec{b}''$
that coincides with $\vec{a} + \vec{u}$ on all the components indexed by~$T$.
Letting $\vec{b}^\dag$ be some vector such that
$\vec{b}^\dag \geq \vec{b''}$ and $\lambda(\vec{b}^\dag) = \lambda(\vec{a} + \vec{u})$,
subsequently $\widehat{\tau}$ plays according to some winning strategy of Player~2
in the energy game $\Delta_{\vec{b}^\dag}\!\tuple{V, E, d}$ from the vertex~$v''$,
which exists by the inductive hypothesis.

\item
The other way in which playing by $\tau$ can reach the vertex~$v_\bot$
is by and edge in $\mathsf{TrLo}_T\!\tuple{V, E, d}$ of the form
$\tuple{\tuple{v', \vec{a}}, \vec{u}|_{\overline{T}}, v_\bot}$ where
there exists an edge $\tuple{v', \vec{u}, v''}$ in $\tuple{V, E, d}$ such that
$\vec{a} + \vec{u}$ has a negative component.
Then $\widehat{\tau}$ follows such an edge $\tuple{v', \vec{u}, v''}$
and at that point wins the energy game for Player~2, so subsequently may play arbitrarily.
\qedhere
\end{itemize}
\end{proof}

\subsubsection{Complexity Upper Bounds.}
\autoref{th:energy.bounding} not only tells us
how to solve energy games for given initial credits using bounding games,
but it also has the following easy corollary saying that
energy games with arbitrary initial credits can be solved
just using the bounding game on the lossy game graph.

\begin{corollary}
\label{cor:energy.arbitrary}
The following are equivalent:
\begin{enumerate}
\item
There exists $\vec{b} \in \+N^d$ such that
Player~1 wins energy game $\Delta_{\vec{b}}\!\tuple{V, E, d}$
from vertex~$v$.
\item
Player~1 wins energy game $\Delta_{\tuple{\lambda_1, \ldots, \lambda_1}}\!\tuple{V, E, d}$
from vertex~$v$.
\item
Player~1 wins bounding game $\Gamma\!\left(\mathsf{TrLo}_{\emptyset}\!\tuple{V, E, d}\right)$
from vertex \\ $\tuple{v, \tuple{\infty, \ldots, \infty}}$.
\end{enumerate}
\end{corollary}

\begin{example}
  By \autoref{cor:energy.arbitrary}, because she was winning the energy game
  of \autoref{fig-wgame} with initial credit $(2,2)$, Player~1 is now
  winning the bounding game played on the lossy \ifshort version of
  the multi-weighted game graph of \autoref{fig-wgame}\else
  multi-weighted game graph of \autoref{fig-lgame}\fi.
\end{example}

\begin{figure}[tbp]
  \centering
  \begin{tikzpicture}[auto,on grid,node distance=3cm]
  \node[triangle](C){$v_0$};
  \node[square,left=of C](L){$v_L$};
  \node[square,right=of C](R){$v_R$};
  \path[->,every node/.style={font=\footnotesize}]
  (C) edge[swap] node{$\tuple{0,0}$} (L)
  (C) edge node{$\tuple{0,0}$} (R)
  (L) edge[bend left=40,draw=black!40!red] node{$\tuple{-2,2}$} (C)
  (L) edge[bend right=26,draw=black!40!orange,swap] node{$\tuple{-1,3}$} (C)
  (R) edge[bend right=40,draw=black!40!blue,swap] node{$\tuple{2,-1}$} (C)
  (R) edge[bend left=26,draw=black!40!cyan] node{$\tuple{3,-3}$} (C)
  (C) edge[loop above] node {$\tuple{-1,0}$} ()
  (C) edge[loop below] node {$\tuple{0,-1}$} ();;
  \end{tikzpicture}
  \caption{A simplified result of the $\mathsf{TrLo}_\emptyset$ construction
           on the game graph of \autoref{fig-wgame}.\label{fig-lgame}}
\end{figure}

It is also straightforward to obtain the next corollary.
Its first part follows from \autoref{cor:energy.arbitrary} and \autoref{cor:haupt}.
Its second and third parts follow from \autoref{th:energy.bounding},
since $O(|V| \cdot \|E\|)^{2^{O(d \cdot \log d)}}$ deterministic time
suffices for solving the games $\Gamma\!\left(\mathsf{TrLo}_{T}\!\tuple{V, E, d}\right)$
for all $T \in \{1, \ldots, d\}$,
and all components of all vectors in the Pareto limit are at most
$\lambda_d = (4 |V| \cdot \|E\|)^{(3 d)^{4 d}}$.

\begin{corollary}
\label{cor:energy.upper}
For energy games on multi-weighted game graphs $\tuple{V, E, d}$:
\begin{enumerate}
\item
the arbitrary initial credit problem is solvable in
$O(|V| \cdot \|E\|)^{O(d^4)}$ deterministic time;
\item
the given initial credit problem is solvable in
$O(|V| \cdot \|E\|)^{2^{O(d \cdot \log d)}}$ deterministic time;
\item
the Pareto limit, consisting of all pointwise minimal winning initial credits, is computable in
$O(|V| \cdot \|E\|)^{2^{O(d \cdot \log d)}}$ deterministic time.
\end{enumerate}
\end{corollary}

The upper bound for the given initial credit problem
matches the \EXP[2] lower bound from \citep{courtois14},
and encompasses \citeauthor{chaloupka13}'s \P\ upper bound
in dimension $d=2$ with unit updates, i.e.\ with $\|E\|=1$.
Because the given initial credit problem is \EXP-hard
for fixed dimension $d\geq 4$~\citep{courtois14},
the upper bound in terms of $\|E\|$ cannot be improved.

\fi


\section{Multi-Dimensional Mean-Payoff Games}
\label{sub:meanpayoff}
This section summarises the technique for solving multi-dimensional
mean-payoff games proposed by \citet{chatterjee13}, which relies
on \emph{open half-spaces}.  The rest of the paper does not rely
formally on this section and it may be omitted by a reader eager to
get on with our new `perfect half-spaces' technique for solving
multi-dimensional bounding games.  We believe, however, that starting
here helps put our work in context, appreciate similarities and
differences between the two techniques, and understand the conceptual
and some of the technical challenges we had to overcome.

\subsection{Multi-Dimensional Mean-Payoff Games}
Given a play $\rho$ over a multi-weighted game graph $\tuple{V,E,d}$,
we define its \emph{long-term average} in $\+Q^d$ as
$\avg{\rho}\eqdef\lim_{n\to\infty}\frac{w(\rho|_n)}{n}$.  We say that
$\rho$ is winning for Player~1 in the \emph{mean-payoff game}
$M\!\tuple{V,E,d}$ if $\avg{\rho}\geq\vec 0$.  Otherwise, i.e.\ if
there is a coordinate $1\leq i\leq d$ such that $\avg{\rho}(i)<0$, the
play $\rho$ is winning for Player~2.  As shown by
\citet{chatterjee10}, determining the winner in multi-dimensional
mean-payoff games is \coNP-complete, and in pseudo-polynomial time
when the dimension is
fixed~\citep[\theoremautorefname~1]{chatterjee13}.

\subsection{Energy Versus Mean-Payoff}
In a one-dimensional ar\-bit\-ra\-ry-initial-credit energy game, the
goal of Player~1 is to keep the energy level bounded from below.   
It is folklore that Player~1 has a winning strategy in such a game if
and only if she has a strategy in the mean-payoff game on the same
game graph that guarantees a non-negative long-term average.

\subsubsection{Infinite Memory Strategies for Player~1.}
This relationship between energy games and mean-payoff games does not
generalise to multi-dimensional games.  We illustrate this on the
example of a 2-dimensional game graph from \autoref{fig:mean}.  In
\autoref{ex:mean} we have argued that Player~2 has a winning strategy
in the bounding game (and hence also in the arbitrary-initial-credit
energy game).  On the other hand, we argue that Player~1 has a
strategy to guarantee that the long-term average in both dimensions is
non-negative.  Indeed, consider a strategy in which in stage~$m$---for
all $m = 1, 2, 3, \dots$---Player~1 performs one of the self-loops
$m$~times, then she moves to the other vertex where she performs the
other self-loop $m$~times, and then finally returns to the starting
vertex.  After $m$ stages, the energy level in both dimensions is~$-m$
and the number of steps performed is~$\Theta(m^2)$, hence the
long-term average in the infinite play is~0 in both dimensions,
because $\lim_{m \to \infty} -\frac{m}{m^2} = 0$.

Note that this strategy for Player~1 in the game graph of
\autoref{fig:mean} is infinite-memory, since the actions depend on the
stage $m$.  Multi-dimensional mean-payoff games might require infinite
memory in order to be won, as shown by
\citet[\lemmaautorefname~4]{chatterjee10}---their proof can be used to
show that the game of \autoref{fig:mean} actually requires infinite
memory.

\subsubsection{Finite Memory Strategies for Player~1.}
In the multi-dimensional case, there is nevertheless a strong relation
between energy and mean-payoff games.  Call a strategy $\sigma$
\emph{finite memory} if there exists an equivalence relation $\sim$
with finite index over $(V\times\+Z^d)^+$ such that, whenever
$\pi\sim\pi'$ for some non-empty paths $\pi$ and $\pi'$ in the domain
of $\sigma$, then $\sigma(\pi)=\sigma(\pi')$ (such strategies are
typically described using \emph{Moore machines}).
\begin{fact}[\citet{chatterjee10}]
  Let $\tuple{V,E,d}$ be a multi-weighted game graph.  There exists an
  initial credit $\vec b$ such that Player~1 wins the energy game
  $\Delta_{\vec b}\tuple{V,E,d}$ if and only if Player~1 has a finite
  memory winning strategy in the mean-payoff game $M\!\tuple{V,E,d}$.
\end{fact}
Hence our complexity bounds in \autoref{cor:energy.upper} on
multi-dimensional energy games also yield a pseudo-polynomial time
algorithm to find a winning finite-memory strategy for Player~1 in a
given fixed-dimensional mean-payoff game.

\subsection{The Open Half-Space Technique for Mean-Payoff Games}
For technical convenience, we follow \citeauthor{chatterjee13} in
considering mean-payoff games on lossy game graphs.
In this context, the goal of Player~1 is to achieve a long-term
average of~0 in all dimensions, and the goal of Player~2 is to achieve
a negative long-term average in at least one dimension.

\subsubsection{Winning Strategies for Player~2.}
The first key observation that underpins the solution of (lossy)
multi-dimensional mean-payoff games by \citeauthor{chatterjee13} is
the following sufficient condition for the existence of a winning
strategy for Player~2 from some vertex in the game graph:
there is a vertex~$v_0$, an open half-space 
$H \subseteq {\mathbb R}^d$ and a strategy for Player~2 that  
guarantees all simple cycles formed along a play from $v_0$ to be
in~$H$.
One can then argue that if Player~2 uses such a strategy
indefinitely then the norms of the energy level vectors grow linearly
in the number of steps performed, and hence the long-term average is
non-zero in at least one dimension.

Every open half-space can be determined by a
non-zero vector~$\vec{n} \in \mathbb{R}^d$ that is normal to the
hyperplane on the boundary of the half-space:
$$H_{\vec{n}} = \{\vec{v} \in \mathbb{R}^d \: : \: \vec{n} \cdot
\vec{v} < 0\}\;.$$ \Citet[\lemmaautorefname~1]{chatterjee13} crucially
point out that for every vector~$\vec{n} \in \mathbb{R}^d$, one can
check whether Player~2 has a strategy that guarantees all simple
cycles formed to be in~$H_{\vec{n}}$ from $v_0$ by solving a
one-dimensional mean-payoff game on the multi-weighted game graph with
every weight $\vec u$ replaced by the dot-product~$\vec{n} \cdot \vec
u$.

\subsubsection{Winning Strategies for Player~1.}
The second key insight of \citeauthor{chatterjee13} is that the above
sufficient condition for the existence of a winning strategy for
Player~2 in a lossy multi-dimensional mean-payoff game is 
necessary.  
Indeed, by (positional) determinacy of mean-payoff
games~\cite{ehrenfeucht79}, it follows that, if the sufficient
condition described above does not hold for any open half-space
$H_{\vec n}$ and any initial vertex $v_0$, then for all non-zero
vectors~$\vec{n}$ and all vertices $v_0$, Player~1 has a (positional)
strategy to block simple cycles in~$H_{\vec{n}}$ along any play from
$v_0$, or in other words to force all simple cycles formed to be
in $$\mathbb{R}^d \setminus H_{\vec{n}} = \overline{H_{-\vec{n}}} =
\{\vec{v} \in \mathbb{R}^d \: : \: \vec{n} \cdot \vec{v} \geq 0\}\;.$$

In such a case, \citet[\lemmaautorefname~2]{chatterjee13} show that
such strategies of Player~1, which force simple cycles formed to be in
any closed half-space, can be carefully combined to ensure that the
long-term average is~0 in every dimension.  The main idea in the
construction of the strategy for Player~1 is to proceed in stages
$m=1, 2, 3, \dots$, to monitor the energy-level vector at the
beginning of stage~$m$ of the game, say $\vec{g}_m$, and to
`counteract' its further growth in the direction of~$\vec{g}_m$
throughout stage~$m$ by using the strategy that blocks simple 
cycles in the open half-space $H_{\vec{g}_m}$, i.e., that forces all
the formed simple cycles to be in the closed half-space
$\overline{H_{-\vec{g}_m}}$. 
In the winning strategy we described for Player~1 for the mean-payoff
game over the graph of \autoref{fig:mean}, Player~1 can avoid cycles
in $H_{(-1,1)}$ by playing the self-loop on $v_L$, and she can avoid
cycles in $H_{(1,-1)}$ by playing the self-loop on $v_R$.

Moving from one stage to another, and hence switching between such
counteracting strategies to force simple cycles in different
half-spaces, cannot be done too often because as a result of switching
from one strategy to another a bounded number of unfavourable simple
cycles may be formed.  This is the case in our example, since
switching between $v_L$ and $v_R$ closes a cycle with effect $(-1,-1)$
resulting in a drift away from the non-negative orthant.  

The strategy for Player~1 proposed by \citeauthor{chatterjee13}
overcomes this complication by 
increasing the number of steps made in every stage;
in particular, they proposed making $s(m)\eqdef m$ steps in stage~$m$
before proceeding to stage~$m+1$.
The purpose is to make the drift grow slower than the number of steps
in the play. 
This, as can be deduced from their 
analysis, gives a bound of $O(n^{3/4})$ for the norm of the
energy-level vector after~$n$ steps, and hence the long-term average
is~0 in all dimensions because 
$\lim_{n \to \infty} \frac{n^{3/4}}{n}=0$.  
One may observe that more 
generally, if we set $s(m)\eqdef m^\varepsilon$, for any
$\varepsilon>0$, then the norm of the energy-level vector after~$n$
steps is $O(n^{1/2 + \varepsilon/2(1+\varepsilon)})$.  Hence, the
best upper bounds on the norm of the energy-level vectors after $m$
steps that can be guaranteed by Player~1---when using a strategy
similar to that constructed by \citeauthor{chatterjee13}---are in
$\omega(\sqrt{n})$.  Let us point out that such strategies require
infinite memory because they need to `keep the count' of the stage
they are in and of the number of steps they need to perform in the
current stage, both of which are unbounded.

\section{Perfect Half-Spaces}
\label{sec-pphs}
We recall in this section the definition of subsets of $\+Q^d$
called \emph{perfect half-spaces}\ifshort\else, which can also
be characterised as the \emph{maximal} salient blunt cones in
$\+Q^d$\fi.  They will be used next in \autoref{sec-fcycle} to define
a condition for Player~2 to win bounding games,
which relies on Player~2's ability to force
cycles 
inside perfect
half-spaces.  This can be understood as a generalisation
of \citeauthor{chatterjee13}'s approach for solving
multi-dimensional \ifshort
\else mean-payoff games, which \fi relies
on a similar ability to force cycles inside open half-spaces.
We employ perfect half-spaces in \autoref{sec-fcycle} to colour the
edges in \emph{first-cycle bounding games}, which determine the
winner using both the colours and the weight of the first cycle formed
along a play.


\subsection{Definitions from Linear Algebra}
Given a subset $\vec A$ of $\+Q^d$, we write $\vsp{\vec A}$ (resp.,
$\cone{\vec A}$) for the \emph{vector space} (resp.,
the \emph{cone}) \emph{generated} by $\vec A$, i.e., the closure of
$\vec A$ under addition and under multiplication 
by all (resp., nonnegative) rationals.
\ifshort\else\par
Observe that the sufficient condition for existence of a winning
strategy for Player~2 in a lossy multi-dimensional mean-payoff game is
also a sufficient condition for him to have a winning strategy in a
bounding game.  Unlike for multi-dimensional mean-payoff games and as
witnessed with the game on \autoref{fig:mean}, however, this condition
is not necessary.
In order to formulate a new more powerful sufficient condition, we use
instead perfect half-spaces: a \fi \emph{$k$-perfect
half-space} of $\mathbb{Q}^d$, where $k \in \{1, 2, \dots, d\}$, is a
(necessarily disjoint) union $H_d \cup \cdots \cup H_k$ such that:
\begin{itemize}
\item
  $H_d$ is an open half-space of~$\mathbb{Q}^d$;
\item
  for all $j \in \{k, \dots, d - 1\}$, $H_j \subseteq \mathbb{Q}^d$ 
  is an open half-space of the boundary of~$H_{j + 1}$.
\end{itemize}
Whenever we write a $k$-perfect half-space in form
$H_d \cup \cdots \cup H_k$, we assume that each $H_j$ is
$j$-dimensional.  We additionally define the $(d + 1)$-perfect
half-space as the empty set; a \emph{partially-perfect half-space} is
then a $k$-perfect half-space for some $k$ in $\{1,\dots,d+1\}$.
A \emph{perfect half-space} is a $1$-perfect
half-space.  \ifshort\else Observe that
a partially-perfect half-space is always a cone, which
is \emph{blunt}, i.e., does not contain $\vec 0$,
and \emph{salient}, i.e., if it contains a vector $\vec v$ then it
does not contain its opposite $-\vec v$.  Moreover, a perfect
half-space is a \emph{maximal} blunt and salient cone.\fi

\subsection{Generated Perfect Half-Spaces}
In order to pursue effective and parsimonious strategy constructions,
we consider perfect half-spaces generated by particular sets of
vectors, which will correspond to the total weights of simple cycles
in multi-weighted game graphs.  Given a norm $M$ in $\+N$, we say that
an open half-space $H$ is \emph{$M$-generated} if its boundary equals
$\vsp{\vec B}$ for some set $\vec B$ of vectors of norm at most $M$%
.  By extension, a partially-perfect
half-space is \emph{$M$-generated} if each of its open half-spaces is
$M$-generated.
\begin{proposition}
\label{pr:sib}
  Any $k$-dimensional vector space of $\mathbb{Q}^d$ has
  at most $\?L(k)\eqdef 2 (2 M+ 1)^{d (k - 1)}$ open half-spaces that are
  $M$-generated.
\end{proposition}

\begin{example}\label{ex:pphs}
In the game graph of \autoref{fig:mean}, there are three $1$-generated
open half-spaces of interest: the half-plane
$H_2\eqdef\{(x,y):x+y<0\}$ with boundary $\vsp{(-1,1),(1,-1)}$ and
containing $(-1,-1)$, and the two half-lines
$H_1\eqdef\{(x,y):x+y=0\wedge x<0\}$ and
$H'_1\eqdef\{(x,y):x+y=0\wedge x>0\}$ with boundary $\vsp{\vec 0}$ and
containing, respectively, $(-1,1)$ and $(1,-1)$.
\ifshort\else In turn, those three \fi open half-spaces define
two perfect half-spaces: $H_2\cup H_1$ and $H_2\cup H'_1$.
\end{example}

\subsection{Hierarchy of Perfect Half-Spaces}
Finally, we fix a ranked tree-like structure on all $M$-generated
partially-perfect half-spaces, which provide a scaffolding on which we
will build strategies in multi-dimensional bounding
games.  Observe 
that an $M$-generated partially-perfect half-space 
$H_d\cup\cdots\cup H_k$ for $k>1$ can be extended using any of the 
$M$-generated open half-spaces $H$ of the boundary of $H_k$; note 
that this boundary then equals $\vsp{H}$.  In \autoref{ex:pphs}, $H_2$
can be extended using $H_1$ or $H'_1$, and
$\vsp{H_1}=\vsp{H'_1}=\{(x,y):x+y=0\}$.

The set of $M$-generated perfect half-spaces can be totally ordered by
positing a linear ordering $<$ between all $M$-generated open
half-spaces.  We write $\prec$ for the lexicographically induced
linear ordering between all $M$-generated perfect half-spaces of
$\mathbb{Q}^d$\/: if $\mathcal{H} = H_d \cup \dots \cup H_1$ and
$\mathcal{H'} = H'_d \cup \dots \cup H'_1$, we define
$\mathcal{H} \prec \mathcal{H'}$ to hold iff $H_j = H'_j$ for all
$j \in \{k+1, \dots, d\}$ and $H_k < H'_k$ for some $k \in \{1,
2, \dots, d\}$.

\section{First Cycle Bounding Games}
\label{sec-fcycle}
We define in this section \emph{first-cycle bounding games}, which
provide the key technical arguments for most of our results.  Such
games end as soon as a cycle is formed along a play, and the weight of
this cycle determines the winner, along with a colouring information
chosen by Player~2.  In sections~\ref{sub-player2}
and~\ref{sub-player1}, we are going to show that first-cycle
bounding games and infinite bounding games are equivalent, by
translating winning strategies for each Player~$p$, $p\in\{2,1\}$,
from first-cycle bounding games to bounding games.  This yields in
particular the small hypercube property of \autoref{lem-shp}.

\subsection{Definition} 
We define 
the \emph{first-cycle bounding game}
$G\!\tuple{V, E, d}$ on a multi-weighted game graph 
$\tuple{V, E, d}$:
\begin{itemize}
\item
  at any Player-$1$ vertex, Player~$2$ chooses a
  $|V|\cdot \|E\|$-generated perfect half-space $\?H$ of
  $\mathbb{Q}^d$, and then Player~$1$ chooses an outgoing edge, whose
  occurrence in the play becomes coloured by~$\?H$;
\item
  at any Player-$2$ vertex, he chooses an outgoing edge;
\item
  the game finishes as soon as a vertex is visited twice,
  which produces a simple cycle $C$ with coloured Player-$1$ edges;
\item
  Player~$2$ wins if $w(C)$, the total weight of the cycle, is in the
  largest partially-perfect half-space of $\mathbb{Q}^d$ that is
  contained in all the colours in~$C$, i.e.\ the least
  common ancestor of all the colours in~$C$;
  Player~$1$ wins otherwise.
\end{itemize}

\begin{example}
  Player~2 wins the first-cycle bounding game played
  in \autoref{fig-wgame} (but loses in its lossy version).
  E.g.\ strategy~\eqref{eq-strategy2} is winning for Player~2 if he
  colours the edges outgoing from $v_0$ by the perfect half-space
  $H'_2\cup H_1$ where $H'_2\eqdef\{(x,y):x+y>0\}$ and
  $H_1\eqdef\{(x,y):x+y=0\wedge x<0\}$.
\end{example}

\begin{example}
  Player~2 wins the first-cycle bounding game played
  in \autoref{fig:mean}.  Indeed, he can choose the colour $H_2\cup
  H_1$ in $v_L$ and the colour $H_2\cup H_1'$ in $v_R$.  Then Player~1
  cannot avoid forming a simple cycle in either $H_2\cup H_1$ (if
  cycling on $v_L$), in $H_2\cup H_1'$ (if cycling on $v_R$), or in
  $H_2$ (if cycling between $v_L$ and $v_R$).
\end{example}
Observe that first-cycle bounding games are finite perfect information
games, and are thus \emph{determined}: from any vertex, either
Player~1 wins or Player~2 wins.

\subsection{Winning Strategies for Player~2}\label{sub-player2}
Suppose $\sigma$ is a strategy of Player~$2$ from a vertex $v_0$ in a
first-cycle bounding game $G\!\tuple{V, E, d}$.  Let
$\widetilde{\sigma}$ be the following strategy of Player~$2$ in the
infinite bounding game $\Gamma\!\tuple{V, E, d}$:
\begin{itemize}
\item
at any Player-$2$ vertex, $\widetilde{\sigma}$ chooses 
the edge specified by~$\sigma$;
\item
whenever a cycle is formed, $\widetilde{\sigma}$ cuts it out of its memory,
and continues playing according to~$\sigma$.
\end{itemize}

\begin{restatable}{lemma}{claimplayertwo}
\label{l:Exp}
If $\sigma$ is winning for Player~2 in $G\!\tuple{V, E, d}$ from some
vertex $v_0$, then $\widetilde{\sigma}$ is winning for Player~2 in
$\Gamma\!\tuple{V, E, d}$ from the same vertex $v_0$.
\end{restatable}
\ifshort\pagebreak\begin{proof}[Proof idea]\else\begin{proof}\fi
Consider any infinite play $\widetilde\rho$ consistent with
$\widetilde{\sigma}$, and let:
\begin{itemize}
\item
$\rho$ be obtained from $\widetilde\rho$ by colouring all Player~$1$'s
  edges with the $|V| \cdot \|E\|$-generated perfect half-spaces of
  $\mathbb{Q}^d$ as specified by~$\sigma$;
\item
$C_1, C_2, \ldots$ be the cycle decomposition of~$\rho$,
and for each $n$, $\rho_n$ be the simple path 
that remains after removing~$C_n$;
\item
$\?H_n$ be the largest partially-perfect half-space of $\mathbb{Q}^d$
  that is contained in all the colours in $C_n$, for each~$n$.
\end{itemize}
Since $\sigma$ is winning for Player~2 in the first-cycle game, each
cycle weight $w(C_n)$ belongs to the partially-perfect half-space
$\?H_n$.  The bulk of the proof consists in extracting a `direction of
divergence' of the total energy, notwithstanding that the $\?H_n$'s may
keep varying.

In short, by distinguishing those $n$'s for which the length of the
simple path $\rho_n$ is the smallest one that occurs infinitely often,
we are \ifshort\else going \fi to show that the set of $\?H_n$'s
that occur infinitely often has a unique smallest element
$\?H=H_d\cup\cdots\cup H_k$ with respect to inclusion.  Further
linear-algebraic reasoning \ifshort\else in the
upcoming \autoref{l:unb} \fi then shows that one of the component
half-spaces $H_{k'}$ of $\?H$ provides the desired direction of
divergence: after some $N>0$, all the sums of cycle weights
$w(C_N)+w(C_{N+1})+\cdots+ w(C_n)$ belong to the topological closure
$\overline{H_{k'}}$ and their distances from the boundary of $H_{k'}$
diverge.
\ifshort See \appref{app-player2} for
  details.\end{proof}\else%

In more details now, along the infinite play $\rho$, the prefixes
$\rho_n$ might get shorter or longer but are always of length bounded
by $|V|$.  Those lengths are traced in blue in \autoref{fig-rhon}.  We
let $\ell$ be the minimal such length that occurs infinitely often.

Let us call a partially-perfect half-space that occurs infinitely
often in the sequence $\?H_1, \?H_2, \ldots$ a \emph{recurring} one.
We want to show that, among the recurring partially-perfect
half-spaces, there is one that is contained in all the others; the
subsequent \autoref{l:unb} will then allow to conclude.
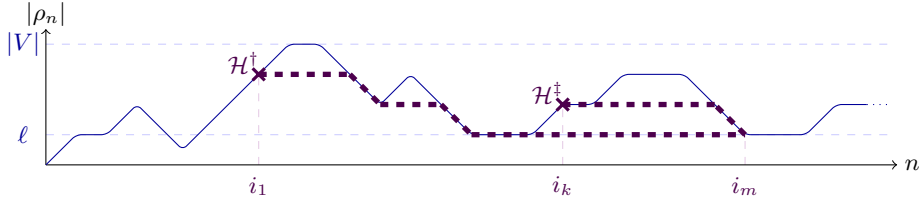
\begin{figure}
  \center
  \begin{tikzpicture}[scale=.4,every node/.style={font=\footnotesize}]
    \draw[color=blue!20,dashed] (0,4) -- (27.8,4);
    \draw[color=blue!20,dashed](0,1) -- (27.8,1);
    \draw[color=violet!20,dashed] (7,0) -- (7,3) (17,0) -- (17,2)
    (23,0) -- (23,1);
    \draw[->] (0,0) -- (28,0);
    \draw[->] (0,0) -- (0,4.5);
    \node at (0,5) {$|\rho_n|$};
    \node at (28.5,0) {$n$};
    \draw[color=black!40!blue,rounded corners=2pt]
      (0,0) -- (1,1) -- (2,1) -- (3,2) -- (4.5,0.5) -- (5,1)
            -- (8,4) -- (9,4) -- (11,2) -- (12,3) -- (14,1)
            -- (16,1) -- (17,2) -- (18,2) -- (19,3) -- (20,3)
            -- (21,3) -- (23,1) -- (25,1) -- (26,2) -- (27,2);
    \draw[color=black!40!blue,dotted] (27,2) -- (27.8,2);
    \node[color=black!40!blue] at (-.75,4) {$|V|$};
    \node[color=black!40!blue] at (-.75,1) {$\ell$};
    \draw[color=black!40!violet,line width=2pt,rounded corners=2pt,dashed] 
               (10,3) -- (11,2)    (13,2) -- (14,1)    (16,1)    (23,1) -- (22,2)    (18,2);
    \draw[color=black!40!violet,line width=2pt,rounded corners=2pt,dashed]
      (7,3) -- (10,3)    (11,2) -- (13,2)    (14,1) -- (16,1) -- (23,1)    (22,2) -- (18,2) -- (17,2);
    \draw[color=black!40!violet,very thick] (7.2,3.2) -- (6.8,2.8) (7.2,2.8) -- (6.8,3.2)
      (17.2,2.2) -- (16.8,1.8) (17.2,1.8) -- (16.8,2.2);
    \node[color=black!40!violet] at (6.5,3.4) {$\?H^\dagger$};
    \node[color=black!40!violet] at (16.5,2.4) {$\?H^\ddagger$};
    \node[color=black!40!violet] at (7,-.7) {$i_1$};
    \node[color=black!40!violet] at (17,-.7) {$i_k$};
    \node[color=black!40!violet] at (23,-.7) {$i_m$};
  \end{tikzpicture}
  \caption{Lengths of the simple paths $\rho_n$ in the proof of
    \autoref{l:Exp}.\label{fig-rhon}}
\end{figure}

First observe that, for any $n < n'$ such that $|\rho_n| \geq
|\rho_{n'}|$ and, for all $m \in \{n + 1, \ldots, n' - 1\}$, $|\rho_n|
\leq |\rho_m|$, we have that either $\?H_n$ contains $\?H_{n'}$ or
vice-versa (since $|\rho_n|$ is bounded by $|V|$ this situation must
occur infinitely often).  Indeed, in this case $C_{n}$ starts with
(and $\rho_{n}$ ends with) a vertex $v_1$ that occurs in $C_{n'}$.
Consider the sequence of vertices $v_1,\ldots,v_{k-1},v_k$ visited
along the cycle $C_{n}$ until the first Player~1 vertex $v_k$ (recall
that by assumption all the simple cycles must visit some Player~1
vertex).  Because in step~$n'$ Player~2 plays according to $\sigma$,
he will choose the same actions in Player~2 vertices
$v_1,\dots,v_{k-1}$ as in step~$n$.  Hence $v_k$ is also visited
inside $C_{n'}$, and Player~2 uses the same edge colouring at
steps~$n$ and~$n'$ in $v_k$.  Thus the two cycles $C_n$ and $C_{n'}$
share a colour, and the corresponding partially perfect half-spaces
$\?H_{n}$ and $\?H_{n'}$ must be comparable for inclusion.

Assume now that there are two incomparable (for inclusion) recurring
partially perfect half-spaces $\?H^\dagger$ and $\?H^\ddagger$ and let
us show that some partially perfect half-space $\?H\subseteq
\?H^\dagger,\?H^\ddagger$ must also be recurring.  Consider the
infinite suffix of the play where $\ell$ is the minimal observed
length.  By the previous observation, between any two occurrences of
$\?H^\dagger$ and $\?H^\ddagger$ in this suffix, we can find a
sequence
$\?H^\dagger=\?H_{i_1},\?H_{i_2},\dots,\?H_{i_k}=\?H^\ddagger$ of
comparable partially perfect half-spaces connecting the two, i.e.,
with $\?H_{i_j}\subseteq\?H_{i_{j+1}}$ or
$\?H_{i_j}\supseteq\?H_{i_{j+1}}$ for all $j$.  This is because there
will be a later occurrence of a simple path $\rho_{i_m}$ of length
$\ell\leq\min(|\rho_{i_1}|,|\rho_{i_k}|)$ for some $i_m\geq i_1,i_k$;
see the thick dashed violet line in \autoref{fig-rhon}.  Then some
$\?H\subseteq \?H^\dagger,\?H^\ddagger$ occurs among those
$\?H_{i_1},\dots,\?H_{i_k}$.  Since there are infinitely many such
pairs of occurrences of $\?H^\dagger$ and $\?H^\ddagger$ but finitely
many $|V| \cdot \|E\|$-generated partially-perfect half-spaces, there
must be infinitely many occurrences of one such $\?H$.

Applying the previous reasoning to every pair of recurring partially
perfect half-spaces $\?H^\dagger$ and $\?H^\ddagger$, we see that
there must be a recurring partially perfect half-space that is
contained in all the others.  We conclude the proof using the
following \autoref{l:unb}.
\end{proof}

\begin{claim}
\label{l:unb}
Suppose $\?H = H_d \cup \cdots \cup H_k$ is 
a partially-perfect half-space of $\mathbb{Q}^d$ and 
$\vec{a}_1, \vec{a}_2, \ldots$ is an infinite sequence of vectors such that:
\begin{itemize}
\item
the set $\{\vec{a}_1, \vec{a}_2, \ldots\}$ is finite;
\item
for each $n$, there exists a partially-perfect half-space of $\mathbb{Q}^d$ 
that contains $\?H$ and~$\vec{a}_n$;
\item
we have $\vec{a}_n \in \?H$ for infinitely many~$n$.
\end{itemize}
Then there exist $k' \in \{d, \ldots, k\}$ and $N>0$ such that
\begin{itemize}
\item
for each $n \geq N$, 
$\vec{a}_n$ belongs to $\overline{H_{k'}}$ the
topological closure of $H_{k'}$, and
\item
the set of all distances of $\vec{a}_N + \cdots + \vec{a}_n$
from the boundary of $H_{k'}$ is unbounded.
\end{itemize}
In particular, the set of all norms 
$\|\vec{a}_1 + \cdots + \vec{a}_n\|$ is unbounded.
\end{claim}

\begin{proof}
We have that $\?H$ is of the form $H_d \cup \cdots \cup H_k$.
Let $k' \in \{d, \ldots, k\}$ be maximal such that 
$\vec{a}_n \in H_{k'}$ for infinitely many $n$.
Observe that $\{1, 2, \ldots\}$, the set of all positive integers,
can be partitioned into three:
\begin{enumerate}
\item
The set of all $n$ such that $\vec{a}_n \,\in\, H_d \cup \cdots \cup
H_{k' + 1}$, which is finite by definition of $k'$.  We let $N$ be
larger than the index of the last such~$\vec a_n$; then the vectors
$\vec a_n$ for $n\geq N$ belong to $\overline{H_{k'}}$.
\item
The set of all $n$ such that $\vec{a}_n \in H_{k'}$, which is infinite
by definition of $k'$.  Since $H_{k'}$ is open and the set
$\{\vec{a}_1, \vec{a}_2, \ldots\}$ is finite, there is a positive
minimal distance of those $\vec{a}_n$ from the boundary of~$H_{k'}$.
These vectors bring the sums $\vec a_N+\cdots+\vec a_n$ for $n\geq N$
unboundedly far from the boundary of $H_{k'}$. 
\item
The set of all $n$ such that $\vec{a}_n$ is contained in the boundary
of~$H_{k'}$.  These vectors have no effect on the distance between the
sums $\vec a_N+\cdots+\vec a_n$ for $n\geq N$ and the boundary of
$H_{k'}$.\qedhere
\end{enumerate}
\end{proof}

\fi

\subsection{Winning Strategies for Player~1}\label{sub-player1}

If there is no winning strategy for Player~2 in the first-cycle
bounding game $G\!\tuple{V, E, d}$ from a vertex~$v_0$, 
then by determinacy of first-cycle bounding games, there is a winning
strategy~$\sigma$ for Player~1 in $G\!\tuple{V, E, d}$ from~$v_0$
.

\begin{example}\renewcommand{\theenumi}{\roman{enumi}}\renewcommand{\labelenumi}{(\theenumi)}
  \ifshort
\else Recall the lossy game graph
  from~\autoref{fig-lgame}.  \fi Because Player~1 wins the energy game
  with initial credit $(2,2)$, by \autoref{cor:energy.arbitrary}
  and \autoref{l:Exp}, she wins the first-cycle bounding game.  One
  winning strategy, whose moves depend only on the latest visited
  vertex (here only $v_0$) and colour $\?H$ chosen by Player~2 in
  $v_0$, is as follows: \begin{enumerate}\item\label{ex-colour1} if
  $(-2,2)$ and $(-1,3)$ are both outside $\?H$, move to $v_L$,
  and \item\label{ex-colour2} if $(2,-1)$ and $(3,-3)$ are both
  outside $\?H$, move to $v_R$, and \item\label{ex-colour3} otherwise
  perform the self-loop labelled $(-1,0)$.\end{enumerate} Observe
  that the first two cases~\eqref{ex-colour1} and~\eqref{ex-colour2}
  are disjoint.  Since there is no perfect half-space that contains
  $(-1,0)$ and intersects both $\{(-2,2),(-1,3)\}$ and
  $\{(2,-1),(3,-3)\}$, this strategy is indeed winning for
  Player~1---the same would apply if she were to choose the other
  self-loop $(0,-1)$ instead.
\end{example}

The proof of our main result consists in constructing from $\sigma$ a
finite-memory winning strategy $\widetilde{\sigma}$ for Player~1 in
the infinite bounding game $\Gamma\!\tuple{V, E, d}$ from $v_0$, which
\ifshort\else balances her various `perfect half-space avoidance strategies' in
order to ensure \fi the small hypercube property stated
in \autoref{lem-shp}.
Let us outline this construction.  The memory of $\widetilde{\sigma}$
consists of:
\begin{description}
\item[a simple path]
   $\gamma$ from the initial vertex $v_0$ to the current
  vertex~$v$, in which Player~$1$'s edges are coloured by 
  $|V| \cdot \|E\|$-generated perfect half-spaces
  of~$\mathbb{Q}^d$ (this can be represented concretely by a sequence
  of coloured edges from $E$);
\item[a colour]
  i.e.\ a $|V| \cdot \|E\|$-generated perfect half-space 
  $\mathcal{H} = H_d \cup \cdots \cup H_1$ of~$\mathbb{Q}^d$
  (initially the $\prec$-minimal one);  
\item[counters] $\mathsf{c}(k, W)$
  for every $k \in \{1, 2, \dots, d\}$
  and for every nonzero total weight $W$ of a simple cycle, which are
  natural numbers (initially $0$).
\end{description}
Strategy~$\widetilde{\sigma}$ copies its moves from strategy~$\sigma$
for the first-cycle bounding game, based on the coloured simple path
and the colour it has in its memory.  Whenever a cycle is formed it is
removed from the simple path, and provided its weight $W$ is nonzero,
all the counters $\mathsf{c}(k,W)$ are incremented.

Together with the current path, the counters provide the current
energy level, which equals $w(\gamma)+\sum_W \mathsf{c}(d, W) \cdot W$
throughout the play, where $W$ ranges over all simple cycle weights.
To keep the counters and thus the total energy bounded,
$\widetilde\sigma$ may perform one of the following operations after a
counter increment:
\begin{itemize}
\item
a \emph{$k$-shift to $H'_k>H_k$} changes the current colour $\?H$ to
the $\prec$-minimal perfect half-space of the form $H_d\cup\cdots\cup
H_{k+1}\cup H'_k\cup\cdots\cup H'_1$, and resets to $0$ all the counters
$\mathsf{c}(k',W)$ with $k'<k$;
\item 
a \emph{$k$-cancellation} changes the current colour $\?H$ to the
$\prec$-minimal perfect half-space of the form $H_d\cup\cdots\cup
H_{k+1}\cup H'_k\cup\cdots\cup H'_1$.  Simultaneously, given some
simple cycle weights $W_1,\ldots,W_n$ and a positive integral
solution $\vec x$ to \ifshort$\else
\begin{equation}\label{eq-system}\fi\sum_{i=1}^n\vec x(i) W_i=\vec
0\ifshort$, \else\;,\end{equation}\fi it subtracts
$\vec x\cdot u(k)$
where \ifshort$\else\begin{equation}\fi u(k)\eqdef
(4|V| \cdot \|E\|)^{(2k-1)(d+2)^2} \ifshort$\else\end{equation}\fi
from all the tuples
$\tuple{\mathsf{c}(k',W_1),\ldots,\mathsf{c}(k',W_n)}$ with $k'\geq
k$, and resets to $0$ all the counters $\mathsf{c}(k',W)$ with $k'<k$.
\end{itemize}
\ifshort\else
These two operations define the main phases of the strategy
$\widetilde{\sigma}$.  A \emph{$k$-event} is either a $k$-shift or a
$k$-cancellation.
By a \emph{$k$-month} we mean a maximal period with only
${<}k$-events.  By a \emph{$k$-year} we mean a maximal period with
only ${<}k$-cancellations and ${\leq}k$-shifts.  This hierarchy of
$k$-events mirrors in some sense the hierarchical structure of
$|V|\cdot\|E\|$-generated perfect half-spaces.
\fi

These operations allow to maintain two main invariants, from which the
small hypercube property of \autoref{lem-p1} is
derived\ifshort\else\ (see \autoref{cl-player1})\fi.  For all
$1\leq k\leq d$ and simple path weights $W$ in the span of $H_k$:
\begin{itemize}
\item initially, after any ${>}k$-shift, and after any ${\geq}k$-cancellation, \ifshort$\else\begin{equation}\label{eq-ksoft}\fi
\mathsf{c}(k,W)<\?U(k)\eqdef(4|V| \cdot \|E\|)^{2k(d+2)^2}
\ifshort$, \else\end{equation}\fi the so-called \emph{$k$-soft
bound};
\item at all times, \ifshort$\else\begin{equation}\label{eq-khard}\fi
\mathsf c(k,W)<\?U(k)+u(k)\ifshort$, \else\end{equation}\fi the
so-called \emph{$k$-hard bound}.
\end{itemize}

To ensure those invariants, strategy~$\widetilde\sigma$ further
maintains that, whenever $\mathsf c(k,W)\geq \?U(k)$ and $W$ is in
$\vsp{H_k}$, then $W$ is in $\overline{H_k}$.  When this new invariant
cannot be preserved by any $k$-shift, then a version of the
Farkas-Minkowski-Weyl Theorem implies that it can be enforced through
a $k$-cancellation, in which a small positive integral solution can be
found for \ifshort
\else \eqref{eq-system} \fi where $W_1,\dots,W_n$ are the offending
cycle weights.

This strategy shows a statement dual to \autoref{l:Exp}, and thereby
entails both the equivalence of infinite bounding games with
first-cycle bounding games and the small hypercube property
of \autoref{lem-shp}:
\begin{restatable}{lemma}{claimplayerone}\label{lem-p1}
  If $\sigma$ is winning for Player~1 in $G\!\tuple{V, E, d}$ from
  some vertex $v_0$, then $\widetilde{\sigma}$ is winning for Player~1
  in $\Gamma\!\tuple{V, W, d}$ from $v_0$, and ensures energy levels
  of norm at most $(4 |V| \cdot \|E\|)^{2 (d + 2)^3}$.
\end{restatable}
\ifshort\else\ifshort This section presents the proof of the following main lemma:
\claimplayerone*
\fi
\subsubsection{$\widetilde{\sigma}$ Summarised.}
Let us first summarise the definition of $\widetilde{\sigma}$.  At any
Player-$1$ vertex, $\widetilde{\sigma}$ chooses the edge that $\sigma$
specifies for history $\gamma$ and perfect half-space~$\mathcal{H}$.
After any move that leads to a vertex not occurring in $\gamma$, the
memory of $\widetilde{\sigma}$ is updated only by extending~$\gamma$.
Otherwise, a cycle~$C$ is formed, and the memory is updated as
follows:
\begin{itemize}
\item
  Cycle $C$ is cut out of~$\gamma$.
  For all $k \in \{1, 2, \dots, d\}$, 
  counters $\mathsf{c}(k, w(C))$ are incremented, unless 
  $w(C) = \vec{0}$.  

\item
  If all soft upper bounds hold, that is if for all 
  $k \in \{1, 2, \dots, d\}$ and all simple-cycle weights  
  $W \in \widehat{H_k}$ we have
  $\mathsf{c}(k, W) < \mathcal{U}(k)$, 
  then the memory update is finished.

\item
  Otherwise, let $k \in \{1, 2, \dots, d\}$ be the largest for which 
  the $k$-soft upper bound $\mathsf{c}(k, W) < \mathcal{U}(k)$ 
  fails for some $W \in \widehat{H_k}$. 

\item
  ($k$-shift) If there is a $|V| \cdot \|E\|$-generated open
  half-space $H$ of $\vsp{H_k}$ such that the $k$-soft upper bound
  holds for all simple-cycle weights $W \in \widehat{H}$, then
  denoting by $H'_k$ the $<$-minimal such $H$, $\mathcal{H}$ is
  replaced by the $\prec$-minimal perfect half-space of form
  $H_d\cup\cdots\cup H_{k+1}\cup H'_k\cup\cdots\cup H'_1$.  All
  counters $\mathsf{c}(k', W)$, where $k' \in \{1, 2, \dots, k-1\}$
  and $W$ is a simple-cycle weight, are reset to~0.

\item
  ($k$-cancellation)
  Otherwise, let $W_1, W_2, \dots, W_n$ be all the non-zero
  simple-cycle weights in~$\vsp{H_k}$ that fail the $k$-soft upper
  bound, and let~$\vec A$ be the matrix whose columns are the 
  vectors~$W_1, W_2, \dots, W_n$. 
  Then by duality and existence of small positive integer solutions of
  systems of linear equations 
  (\autoref{l:alt}, \autoref{l:small} and \autoref{pr:ssw}), 
  it follows that $\vec A \vec{x} = \vec{0}$ has 
  a solution in positive integers bounded by $\mathcal{S}(k)\eqdef (2 (|V|
  \cdot \|E\| + 1))^{(k + 2)^2}$.  The perfect half-space
  $\mathcal{H}$ is replaced by the $\prec$-minimal $|V| \cdot
  \|E\|$-generated perfect half-space of the form $H_d\cup\cdots\cup
  H_{k+1}\cup H'_k\cup\cdots\cup H'_1$.  For every~$W_i$ and every $k' \in \{k, k+1, \dots, d\}$, the
  value of $\mathsf{c}(k', W_i)$ is replaced by \begin{equation}\label{eq-wf}\mathsf{c}(k', W_i) -
  u(k) \cdot \vec{x}(i)\;.\end{equation}  All counters $\mathsf{c}(k', W)$, where $k'
  \in \{1, 2, \dots, k-1\}$ and $W$ is a simple-cycle weight, are
  reset to~0.  A $k$-cancellation is \emph{well-defined} if all the
  differences in \eqref{eq-wf} are non-negative.
\end{itemize}

We prove \autoref{lem-p1} through a sequence of claims.  The first
claim shows that $k$-cancellations are always well-defined:
\begin{claim}\label{cl-cancel}
  Every $k$-cancellation is well-defined.
\end{claim}
\begin{proof}
  We need to show that, just before the $k$-cancellation, for every
  cycle weight $W_i$ in $\vsp{H_k}$ that fails the $k$-soft upper
  bound, and all $k'\geq k$, $\mathsf{c}(k',W_i)\geq u(k)\cdot\vec x(i)$.
  Indeed,
  \begin{align*}
    \mathsf{c}(k',W_i)&\geq  \mathsf{c}(k,W_i)&&\text{by monotonicity
      of the counters for $k$,}\\
    &\geq\?U(k)&&\text{since $W_i$ fails the $k$-soft upper bound,}\\
    &=(4|V|\cdot\|E\|)^{(2k-1)(d+2)^2}\!\!\!\!\!\!\!\!\!&&\cdot(4|V|\cdot\|E\|)^{(d+2)^2}\\
    &\geq u(k)\cdot\?S(k)&&\text{since $\?S(k)\eqdef(2(|V|\cdot\|E\|+1))^{(k+2)^2}$,}\\
    &\geq u(k)\cdot\vec x(i)&&\text{by \autoref{l:small}.}\qedhere
  \end{align*}
\end{proof}

As explained before, the $k$-soft bound $\?U(k)$ in~\eqref{eq-ksoft}
is employed by $\widetilde{\sigma}$ to trigger a $k$-event and a
change of strategy to avoid cycles with weight inside some perfect
half-spaces.  However, this change of strategy might allow a few more
instances of those cycles to be formed---but, crucially, no more than
$u(k)$ further instances.  The $k$-hard bound in~\eqref{eq-khard} is
therefore enforced.  

This informal argument is proven formally in the
following \autoref{cl-player1}.  It entails in particular that the
$d$-hard bound is \emph{always} enforced, since $\vsp{H_d}$ is the
whole space $\+Q^d$:
\begin{claim}\label{cl-player1}
  For all $k \in \{1, 2, \dots, d\}$ and for all cycle weights~$W$ in
  $\vsp{H_k}$,
  \begin{enumerate}
  \item\label{cl-p1-1}($k$-soft bound)
    at the beginning of every $k$-year,
    $\mathsf{c}(k, W) < \?U(k)$, and
  \item\label{cl-p1-2}($k$-hard bound)
    $\mathsf{c}(k, W) < \mathcal{U}(k) + u(k)$.
  \end{enumerate}
\end{claim}
\begin{proof}
  We prove the two statements by nested induction, first on $k$ and
  second on the sequence of $k$-years seen so far.

  Let us start with~\eqref{cl-p1-1}.  For the initial $k$-year, and
  for $k$-years that begin just after a ${>}k$-shift or a
  ${>}k$-cancellation, since then $\mathsf{c}(k,W)=0$,
  \eqref{cl-p1-1} holds trivially.  We are left with the case of
  a $k$-year that begins just after a $k$-cancellation.  We can assume
  using the secondary induction hypothesis that~\eqref{cl-p1-2} holds
  at the end of the previous $k$-year.  Consider then some $W_i$ in
  $\vsp{H_k}$ that fails the $k$-soft upper bound just before that
  $k$-cancellation.  At that time, since $\vsp{H_k}$ was not changed
  by the $k$-cancellation, \eqref{cl-p1-2} applies and $\mathsf
  c(k,W_i)<\?U(k)+u(k)$.  Therefore, at the beginning of the $k$-year,
  $\mathsf{c}(k,W_i)<\?U(k)+u(k)-u(k)\cdot\vec x(i)$, and thus
  $\mathsf c(k,W_i)<\?U(k)$ since $\vec x(i)>0$.

  By the secondary induction, it remains to establish~\eqref{cl-p1-2}
  for every $k$-year such that~\eqref{cl-p1-1} held at its
  beginning---this is the heart of the proof.  Let
  $H_k^1,H_k^2,\ldots,H_k^N$ be the sequence of $k$-dimensional open
  half-spaces considered during the $N$ $k$-months spanned by
  the current $k$-year so far, where $N\leq\?L(k)$ by \autoref{pr:sib}.  We know that all these open
  half-spaces define the same vector space
  $\vsp{H_k^1}=\vsp{H_k^2}=\cdots=\vsp{H_k^N}$; let $W$ belong to that
  space.

  If $W$ satisfies the soft bound, there is nothing to be done.
  Otherwise, by the construction of $\widetilde{\sigma}$ and the
  assumption of \eqref{cl-p1-1} at the beginning of the $k$-year,
  there exists a \emph{first} $k$-month in this sequence, say the
  $L$th for some $1\leq L<N$, after which $W$ fails the soft bound
  onward.  Then, during the $k$-months $1,\ldots,L$, $\mathsf
  c(k,W)<\?U(k)$, and for all $n\in\{L+1,\dots,N\}$, we know that $w$
  belongs to the closure $\overline{H_k^n}$.

  Consider the $n$th $k$-month for $n\in\{L+1,\ldots,N\}$ in the
  current $k$-year; we want to bound the increase on $\mathsf{c}(k,W)$
  during that $k$-month.  There are two cases:
  \begin{description}
  \item[If {\boldmath$W\in H_k^n$}] then the $k$-dimensional space in
    the current colour is left unchanged during the $n$th $k$-month.
    In turn, this means that no vertex of the game graph can be
    visited twice during that $k$-month while forming a cycle of
    weight $W$, as otherwise $\sigma$ would allow to form a cycle with
    effect inside $H_d\cup\cdots\cup H^n_k$ and Player~1 would lose.
    Therefore, cycles with weight $W$ that are closed during the $n$th
    $k$-month can only be formed by consuming edges from the simple
    path at the beginning of the $k$-month.  Hence, $\mathsf c(k,W)$
    can be increased by at most $|V|$.
  \item[Otherwise] $W$ belongs to the boundary
    $\overline{H_k^n}\setminus H_k^n$ of $H_k^n$ and thus $k>1$.  In
    this case, during the $n$th $k$-month, $\mathsf c(k,W)$ can only
    be increased by at most the maximal value of $\mathsf c(k-1,W)$
    during the same $k$-month.  This is because $\mathsf c(k-1,W)$ is
    $0$ at the beginning of the $n$th $k$-month, and
    thereafter it can only decrease through ${<}k$-cancellations (or
    that $k$-month would have ended), which decrease $\mathsf
    c(k,W)$ by the same value.  By the main induction hypothesis
    for~\eqref{cl-p1-1} with $W\in
    \vsp{H_{k-1}}\subseteq(\overline{H_k^n}\setminus H_k^n)$, $\mathsf
    c(k-1,W)$ is less than $\?U(k-1)+u(k-1)$.
  \end{description}
  We conclude that, during the current $k$-year, since $N-L\leq\?L(k)$,
  \begin{itemize}
  \item if $k=1$, $\mathsf c(k,W)-\?U(k)$ is less than $\?L(1)\cdot
    |V|=2|V|< u(1)$, and
  \item if $k>1$, 
    \begin{align*}
      \mathsf c(k,W)-\?U(k)&<\?L(k)\cdot\max(|V|,\?U(k-1)+u(k-1))\\
      &<\?L(k)\cdot 2\?U(k-1)\\
      &<(4|V|\cdot\|E\|)^{(d+2)^2}/2\cdot 2(4|V|\cdot\|E\|)^{2(k-1)(d+2)^2}\\
      &= u(k)\;.\qedhere
    \end{align*}
  \end{itemize}
  \end{proof}

\begin{proof}[Proof of \autoref{lem-p1}]
  By the previous claims, $\widetilde{\sigma}$ is winning for Player~1
  in the infinite bounding game, and thanks to the $d$-hard bound, it
  ensures that the norm of the current energy is bounded by
  $\|w(\gamma)\|+\sum_W (\?U(d)+u(d))\cdot \|W\|$ where $W$ ranges
  over the total weights of simple cycles in the game graph, and
  $w(\gamma)$ is the weight of a simple path.  Hence the norms
  $\|w(\gamma)\|$ and $\|W\|$ are bounded by $(|V|\cdot\|E\|)^d$.
  Finally, there are at most $(2(|V|\|E\|)^d+1)^d$ different total
  weights of simple cycles $W$.
\end{proof}
\fi


\section{Concluding Remarks}
\label{sec-concl}
In this paper, we have shown in \autoref{cor:energy.upper}
that fixed-dimensional energy games can be solved
in pseudo-polynomial time, regardless of whether the initial credit is
arbitrary or given.  For the variant with given initial credit, this
closes a large complexity gap between the \TOWER\ upper bounds of
\citet*{brazdil10} and the lower bounds of \citet{courtois14},
and also settles the complexity of simulation problems between VASS
and finite state systems~\citep{courtois14}:
\begin{corollary}
\label{cor-2exp}
  The given initial credit problem for multi-dimensional energy games is
  \EXP[2]-complete, and \EXP-complete in fixed dimension $d\geq 4$.
\end{corollary}

The main direction for extending these results is to consider
a \emph{parity} condition on top of the energy
condition.  \Citet*{abdulla13} show that multi-dimensional energy
parity games with given initial credit are decidable.  They do not
provide any complexity upper bounds---although one might be able to
show \TOWER\ upper bounds from the memory bounds on winning strategies
shown by \citet[\lemmaautorefname~3]{chatterjee14}---, leaving a large
complexity gap with \EXP[2]-hardness.  This gap also impacts the
complexity of \emph{weak simulation} games between VASS and finite
state systems~\citep{abdulla13}.

\ifshort\else\subsection*{Acknowledgements}\fi
The authors thank Dmitry Chistikov for his assistance in proving
\autoref{l:Exp}, the anonymous reviewers for their insightful comments, and
Christoph Haase, J\'er\^ome Leroux, and Claudine Picaronny for helpful
\mbox{discussions}\ifshort\else\ on linear algebra.\fi

\appendix
\section{Linear Algebra}
\label{sec-algebra}
\subsubsection*{An Alternatives Lemma.}

Given a norm $M$ in $\+N$, we write $\mathbb{Z}^\pm_M$ for the set of
integers $\{-M, \ldots, M\}$.
We say that a vector space, cone, or half-space
in $\mathbb{Q}^d$ is \emph{$M$-generated} iff 
it can be generated by vectors in~$(\mathbb{Z}^\pm_M)^d$.  We
use \emph{Weyl's Theorem}:\ifshort\else\footnote{`In these days the angel of topology and 
the devil of abstract algebra fight for the soul of 
each individual mathematical domain.' \cite{Weyl39}}\fi
\begin{theorem}[\citet{weyl50}, \theoremautorefname~1]
\label{th:weyl}
Any $d$-dimensional cone generated by a finite set $\vec A$ in
$\mathbb{Q}^d$ is the intersection of a finite number of closed
half-spaces, where the boundary of each half-space contains $d - 1$
linearly independent vectors from~$\vec A$.
\end{theorem}

\begin{lemma}
\label{l:alt}
Suppose $\vec A \subseteq (\mathbb{Z}^\pm_M)^d$ is contained in an
$M$-generated subspace $S$ of $\mathbb{Q}^d$.  Either $\vec A$ is
contained in some $M$-generated closed half-space of $S$, or
$\tsum_{\vec{a} \in\vec A} \mathbb{Q}_{> 0} \vec{a}$ contains the zero
vector.
\end{lemma}

\begin{proof}
If the $\cone{\vec A}$ is not the whole space $\vsp{\vec A}$,
then by \autoref{th:weyl}, it is contained in an $M$-generated closed
half-space $\overline{H}$ of $\vsp{\vec A}$.  Since $S$ is
$M$-generated, it is easy to obtain from $\overline{H}$ an
$M$-generated closed half-space of $S$ that contains~$\cone{\vec
A}\supseteq\vec A$.

Otherwise, if $\cone{\vec A}$ is the whole space $\vsp{\vec A}$, it
contains in particular the vectors $- \tsum_{\vec{a} \in\vec
A} \vec{a}$ (from $\vsp{\vec A}$) and $\tsum_{\vec{a} \in\vec
A} \vec{a}$ (from $\cone{\vec A}$), and thus $\tsum_{\vec{a} \in\vec
A} \mathbb{Q}_{> 0} \vec{a}$ contains the zero vector.
\end{proof}

\subsubsection*{Small Solutions.}
We also use a lemma that bounds the positive integral solutions on
systems of linear equations.  The lemma is a corollary of the
following result of \citet{vonzurgathen78}:
\begin{theorem}[\citet{vonzurgathen78}]
\label{th:vzgs}
Let $\vec A$, $\vec b$, $\vec C$, $\vec d$ be 
$m \times n$-, $m \times 1$-, $p \times n$-, $p \times 1$-matrices
respectively with integer entries.
The rank of $\vec A$ is $r$, 
and $s$ is the rank of the $(m + p) \times n$-matrix
$\left(\begin{smallmatrix} \vec A \\ \vec C \end{smallmatrix}\right)$.
Let $M$ be an upper bound on the absolute values of 
those $(s - 1) \times (s - 1)$- or $s \times s$-subdeterminants of 
the $(m + p) \times (n + 1)$-matrix
$\left(\begin{smallmatrix} \vec A & \vec b \\ \vec C & \vec d \end{smallmatrix}\right)$,
which are formed with at least $r$ rows from $(\vec A, \vec b)$.
If $\vec A \vec{x} = \vec b$ and $\vec C \vec{x} \geq \vec d$ 
have a common integer solution,
then they have one with coefficients bounded by $(n + 1) M$.
\end{theorem}

\begin{lemma}[Small Solutions Lemma]
\label{l:small}
Suppose $\vec A$ is a $d \times n$-matrix with 
 entries from $\mathbb{Z}^\pm_M$ and mutually distinct columns.
 If $\vec A \vec{x} = \vec{0}$ has a solution in positive rationals,
 then it has a solution in positive integers bounded by 
 $(2 (M + 1))^{(r + 2)^2}$, where $r$ is the rank of~$\vec A$.
\end{lemma}

\begin{proof}
We can assume that $d = r$.
Apply \autoref{th:vzgs} with 
$\vec b$ the $d$-dimensional zero vector,
$\vec C$ the $n$-dimensional identity matrix and
$\vec d$ the $n$-dimensional vector of ones.
The absolute value of any subdeterminant of
$\left(\begin{smallmatrix}
 \vec A           & \vec{0} \\
 \mathrm{Id} & \vec{1}
 \end{smallmatrix}\right)$
is at most $n^{d + 1} M^d$.
Since $n \leq (2M + 1)^d$, we have that
\begin{multline*}
(n + 1) n^{d + 1} M^d \leq
2 (2(M + 1))^d (2(M + 1))^{d (d + 1)} M^d = \\
2^{(d + 1)^2} (M + 1)^{d (d + 3)} \leq
(2 (M + 1))^{(d + 2)^2}.
\qedhere
\end{multline*}
\end{proof}


\nocite{weyl35}
\bibliographystyle{abbrvnat}
\bibliography{journalsabbr,conferences,energy}
  
\end{document}